\documentclass[english,a4paper]{article}
\usepackage{dcolumn}
\usepackage{bm}
\usepackage{graphicx}
\usepackage{amsmath}
\usepackage{amsfonts}
\usepackage{amssymb}
\usepackage{amsthm}
\usepackage{amstext}
\usepackage{amsbsy}
\usepackage{amsopn}
\usepackage{amscd}
\usepackage{amsxtra}

\theoremstyle{plain}
\newtheorem{theorem}{Theorem}
\newtheorem{prop}[theorem]{Proposition}
\newtheorem{cor}[theorem]{Corollary}
\newtheorem{lemma}[theorem]{Lemma}
\theoremstyle{definition}

\theoremstyle{remark}

\newtheorem{example}{Example}

\newcommand{\bqn}{\begin{eqnarray}}
\newcommand{\eqn}{\end{eqnarray}}
\newcommand{\bi}{\begin{itemize}}
\newcommand{\iii}{\item}
\newcommand{\ei}{\end{itemize}}
\newcommand{\ba}{\begin{array}}
\newcommand{\ea}{\end{array}}
\newcommand{\nn}{\nonumber}

\newcommand{\bc}{\begin{cor}}
\newcommand{\ec}{\end{cor}}
\newcommand{\bp}{\begin{prop}}
\newcommand{\ep}{\end{prop}}
\newcommand{\bt}{\begin{theorem}}
\newcommand{\et}{\end{theorem}}

\newcommand{\eps}{\varepsilon}
\newcommand{\al}{\alpha}
\newcommand{\R}{\mathbb{R}}
\newcommand{\C}{\mathbb{C}}
\newcommand{\qq}{q}
\newcommand{\qqin}{\qq_{\mathrm{in}}}
\newcommand{\ppin}{p_{\mathrm{in}}}
\newcommand{\qqfin}{\qq_{\mathrm{fi}}}

\newcommand{\target}{{\cal T}}
\newcommand{\HHH}{{\cal H}}

\begin{document}

\title{Introduction to the Pontryagin Maximum Principle for Quantum Optimal Control}
\author{U. Boscain\footnote{Laboratoire Jacques-Louis Lions, CNRS, Inria, Sorbonne Universit\'e, Universit\'e de Paris, France, boscain@gmail.com}, M. Sigalotti\footnote{Laboratoire Jacques-Louis Lions, CNRS, Inria, Sorbonne Universit\'e, Universit\'e de Paris, France, mario.sigalotti@inria.fr}, D. Sugny\footnote{Laboratoire Interdisciplinaire Carnot de
Bourgogne (ICB), UMR 6303 CNRS-Universit\'e Bourgogne-Franche Comt\'e, 9 Av. A.
Savary, BP 47 870, F-21078 Dijon Cedex, France, dominique.sugny@u-bourgogne.fr}}

\maketitle


\begin{abstract}
Optimal Control Theory is a powerful mathematical tool, which has known a rapid development since the 1950s, mainly for engineering applications. More recently, it has become a widely used method to improve process performance in quantum technologies by means of highly efficient control of quantum dynamics. This tutorial aims at providing an introduction to key concepts of optimal control theory which is accessible to physicists and engineers working in quantum control or in related fields. The different mathematical results are introduced intuitively, before being rigorously stated. This tutorial describes modern aspects of optimal control theory, with a particular focus on the Pontryagin Maximum Principle, which is the main tool for determining open-loop control laws without experimental feedback. The different steps to solve an optimal control problem are discussed, before moving on to more advanced topics such as the existence of optimal solutions or the definition of the different types of extremals, namely normal, abnormal, and singular. The tutorial covers various quantum control issues and describes their mathematical formulation suitable for optimal control. The connection between the Pontryagin Maximum Principle and gradient-based optimization algorithms used for high-dimensional quantum systems is described. The optimal solution of different low-dimensional quantum systems is presented in detail, illustrating how the mathematical tools are applied in a practical way.
\end{abstract}

\maketitle

\tableofcontents

\section{Introduction}

{ Quantum technology aims at developing practical applications based on properties of quantum mechanics~\cite{roadmapQT}. This objective requires precise manipulation of quantum objects by means of external electromagnetic fields. Quantum control encompasses a set of techniques to find the time evolution of control parameters which perform specific tasks in quantum physics~\cite{cat,dalessandro-book,past-present-future,altafini-ticozzi,dong,Koch2016,RMPSTA,RMPstirap,RMPsugny,daems:2013,genov2014,Deffner_2017,Wakamura_2020,brody2015,carlini2006,lloyd2015,Ruschhaupt_2012,lloyd2014}. In recent years, it has naturally become a key tool in the emergent field of quantum technologies~\cite{roadmapQT,wilhelmreview,cat}, with applications ranging from quantum computing~\cite{cat,kosloff2002,geng2016} to quantum sensing~\cite{calarcoreview} and quantum simulation~\cite{quantumsimulation,calarco2011}.

In the majority of quantum control protocols, the control law is computed in an open-loop configuration without experimental feedback. In this context, a powerful tool is Optimal Control Theory (OCT)~\cite{cat} which allows a given process to be carried out, while minimizing a cost (e.g., 
the control time). This approach has key advantages. Its flexibility makes it possible to adapt to experimental constraints or limitations and its optimal character leads to the physical limits of the driven dynamics. OCT can be viewed as a generalization of the classical calculus of variations for problems with dynamical constraints~\cite{liberzon-book}. Its modern version was born with the Pontryagin maximum principle (PMP) in the late 1950s~\cite{pontryaginbook}. Since the pioneering study of Pontryagin and co-workers, OCT has undergone rapid development and is nowadays a recognized field of mathematical research. Recent tools from differential geometry have been applied to control theory, making these methods very effective in dealing with problems of growing complexity. Many reference textbooks have been published on the subject both on mathematical results and engineering applications~\cite{agrachev-book,new-book,bressan-piccoli,schaettler-book,liberzon-book,boscain-book,jurdjevic-book,leemarkusbook,brysonbook,bonnardbook2012}. Originally inspired by problems of space dynamics, OCT was then applied in a wide spectrum of applications such as robotics or economics. OCT was first used for quantum processes~\cite{dahleh1988,KOSLOFF1989} in the context of physical chemistry, the goals being to steer chemical reactions~\cite{ricebook,SOMLOI1993,past-present-future,koch2015,sugny2005} or to control spin dynamics in Nuclear Magnetic Resonance~\cite{conolly,skinner2003application,KHANEJA2005,KOBZAR2012}. A lot of results have recently been established for quantum technologies, as for example the minimum duration to generate high-fidelity quantum gates~\cite{cat}.\\
Two types of approach based on the PMP have been used to solve optimal control problems in low- and high-dimensional systems, respectively. In the first situation called geometric optimal control theory, the equations for optimality are solved by using geometric and analytical tools. The results can be determined 
analytically or at least with a very high numerical precision. The PMP allows to deduce the structure of the optimal solutions and, in some cases, a proof of their global optimality can be established. In this context, a series of low-dimensional quantum control problems has been rigorously solved in recent years for both closed~\cite{alessandro2001,boscain-mason,boscain-nonisotropic,boscain-q1,garon2013,khaneja2001,khaneja2002,hegerfeldt2013,stefanatos2010,chen2016} and open quantum systems~\cite{KhanejaPNAS,khanejathree,lapert2010,bonnard2009,bonnard2012,stefanatos2004,lapert2013,mukherjee2013}. Specific numerical optimization algorithms have been developed and applied to design control fields in larger quantum systems~\cite{KHANEJA2005,koch2012,calarco2011,machnes2011,machnes2018}. Due to the complexity of control landscape, only local optimal solutions are found with this numerical optimal control approach.

However, despite the recent success of quantum optimal control theory, the situation is still not completely satisfactory. The difficulty of the concepts used in this field does not allow a non-expert to understand and apply easily these techniques. The mathematical textbooks use a specialized and sophisticated language, which makes these works difficult to access. Very few basic papers for physicists are available in the literature, while having a minimum grasp of these tools will be an important skill in the future of quantum technologies. The purpose of this tutorial is to  provide an introduction to the core mathematical concepts and tools of OCT in a rigorous but understandable way by physicists and engineers working in quantum control and in related fields. A deep analogy can be carried out between OCT and finding the minima of a real function of several variables. This parallel is used throughout the text to qualitatively describe the key aspects of the PMP. The tutorial is based on an advanced course for PhD students in physics taught at Saarland University in Spring of 2019. It assumes a basic knowledge of standard topics in quantum physics, but also of mathematical techniques such as linear algebra or differential calculus and geometry. Finally, we hope that this paper will give the reader the prerequisites to access a more specialized literature and to apply 
optimal control techniques to their own control problems.

}

\subsection*{Structure of the paper}
{ A tutorial about optimal control is a difficult task because a large number of mathematical results have been obtained and many techniques have been developed over the years for specific applications. Among others, we can distinguish the following problem classes: finite or infinite-dimensional systems, open or closed-loop control, linear or nonlinear dynamical systems, geometric or numerical optimal control, PMP or Hamilton--Jacobi--Bellmann (HJB) approach\dots\ We briefly recall that the HJB method which is the result of the dynamic programming theory leads to necessary and sufficient conditions for optimality in which the optimal cost is solution of a nonlinear partial differential equation~\cite{liberzon-book}. Unfortunately, this equation is generally very difficult to solve numerically.

This means making choices about which topics to include in this paper. We have deliberately selected specific aspects of OCT that are treated rigorously, while others are only briefly mentioned. The choice fell on basic mathematical concepts which are the most useful in quantum control. We limit our focus on the optimal control of open-loop  finite-dimensional system by using the PMP. In particular, we consider only analytical and geometric techniques to solve low-dimensional control problems. To ensure overall consistency and limit the length of the paper, we do not discuss in depth numerical optimization methods and the infinite-dimensional case~\cite{borzi-book,approximate=exact}, which are also key 
issues
in quantum control. In order to connect this tutorial with the current applications of optimal control to high-dimensional quantum system, we describe the link between the PMP and the most current implementation 
of the gradient-based optimization algorithm (the GRAPE algorithm~\cite{KHANEJA2005}). Finally, we stress that  a precise knowledge of the PMP is an essential skill for numerical optimization, and that  the scope of the material of this paper is much broader than the examples presented.

The paper is built on three reading levels. A first level corresponds to the main text and explains the main concepts necessary to describe and apply the theory of optimal control. Some key ideas in optimal control are first introduced qualitatively for a simple quantum system in Sec.~\ref{introqual}. In addition to the two examples solved in Sec.~\ref{secthree} and \ref{sectwo}, the different notions are described rigorously and then systematically illustrated by examples.
This establishes a direct link  between the mathematical concept and its practical application.
  A second reading level is given by footnotes, which recall a mathematical definition or correspond to a more specific comment which can be skipped on a first reading. A final reading level is available in the appendices. These different paragraphs explain in detail the mathematical origin of the theorems used for the controllability and the existence problem and some standard counter-examples or specificities the reader should have in mind. We point out that these sections are not mathematical proofs of theorems, but rather a description of the formalism introduced in a language accessible to a physicist. In order to facilitate the reading of the paper, a list of the main notations used is given in Sec.~\ref{notations} with the place of the text where they are first introduced.

Although the paper is thought for a physics audience and the mathematical details are kept as simple as possible, our objective is to stick to rigorous statements and claims, since 
 this aspect becomes crucial while implementing optimal control ideas in numerical simulations or in experiments.

The paper is organized as follows. We first introduce the main ideas used in optimal control in the case of a simple quantum system in Sec.~\ref{introqual}. We then show how to formulate an optimal control problem from a mathematical point of view in Sec.~\ref{secformu}. Closed and open quantum systems illustrate this discussion. The different steps to solve such a problem are presented in Sec.~\ref{secsteps} by using the analogy with finding a minimum of a function of several variables. The tutorial continues with a point which is crucial, but often overlooked in quantum control studies, namely the existence of optimal solutions. We present in Sec.~\ref{s-existence} some results based on the Filippov test, which is one of the most important techniques to address this question.
The first-order conditions are described in Sec.~\ref{s-first}, with a specific attention on the different types of extremals and on the statement of the PMP. The connection between the PMP and gradient-based optimization algorithms is described in Sec.~\ref{sectiongradient}. Sections~\ref{secthree} and \ref{sectwo} are dedicated to the presentation of two examples in three and two-level quantum systems, respectively. Recent advances in the application of OCT to quantum technologies are briefly described in Sec.~\ref{applioct}, where we mention some of the current directions that are being followed for the development of these techniques. A conclusion is given in Sec.~\ref{conclu}. Mathematical details about the controllability and existence problems are postponed, respectively, to Appendices~\ref{testcontrollability} and \ref{filippovtheorem}.

\section{Introduction to the optimal control concepts: The case of a two-level quantum system}\label{introqual}
In quantum control, a general problem is to prepare a given quantum state by means of a specific time-dependent electromagnetic pulse. This leads to some questions such as
which states can be achieved or
which
shape of control is required to realize this objective. These aspects, which are addressed rigorously in the rest of the tutorial, are first introduced qualitatively in this section.

To this aim, we consider the control
driving a two-level quantum system
from the ground to the excited state. The system is described by a wave function $\psi(t)$ whose dynamics are governed by the Schr\"odinger equation
$$
i\dot{\psi}=
\begin{pmatrix} E_0 & \Omega(t)\\
\Omega^*(t) & E_1\end{pmatrix}\psi,
$$
where units such that $\hbar=1$ have been chosen. The parameters $E_0$ and $E_1$ denote respectively the energies of the ground and excited states, while $\Omega(t)$ corresponds (up to a multiplicative factor) to a complex external field whose real and imaginary parts are, e.g., the components of two orthogonal linearly polarized laser fields. We consider resonant fields for which the carrier frequency $\omega$ of the laser is equal to the energy difference $E_1-E_0$, namely,
\[\Omega(t)=u(t)e^{i(E_1-E_0)t},\]
where the amplitude $u(t)$ represents the control and is assumed to be real.
We now apply a time-dependent change of variables corresponding to the choice of a rotating frame.
The time evolution of $\tilde{\psi}=\Upsilon^{-1}\psi$, with $\Upsilon=\textrm{diag}(e^{-iE_0t},e^{-iE_1t})$, satisfies the differential equation
$$
i\dot{\tilde{\psi}}=
\begin{pmatrix} 0 & u(t)\\
u(t) & 0\end{pmatrix}\tilde{\psi}.
$$
We denote by $c_1=x_1+iy_1$ and $c_2=x_2+iy_2$ the two complex coordinates of $\tilde{\psi}$ in a basis of the Hilbert space $\mathbb{C}^2$ where the indices 1 and 2 correspond respectively to the ground and excited states. Since $\psi$ is a state of norm 1 and $\Upsilon$ a unitary operator, we deduce that $x_1^2+y_1^2+x_2^2+y_2^2=1$. Starting from the state $x_1=1$, the goal of the control is to bring the system to a target for which $x_2^2+y_2^2=1$. The Schr\"odinger equation is equivalent to the following set of equations for the coefficients $x_k$ and $y_k$:
$$
\begin{cases}
\dot{x}_1=u y_2 \\
\dot{y}_1=-u x_2 \\
\dot{x}_2=u y_1 \\
\dot{y}_2=-u x_1.
\end{cases}
$$
Since $u$ is a real control, we immediately see that the first and the last equations are coupled to each other and decoupled from the two others. In other words, the initial state of the dynamics is only connected to states for which $y_1=x_2=0$, i.e., such that $x_1^2+y_2^2=1$. The system thus evolves on a circle. For our control objective, the only interesting states correspond  therefore to $y_2=\pm 1$. It is also straightforward to verify that such target states can be reached at least with a constant control $u$. In control theory, this formulation of the control problem and the analysis of the reachable set from the initial state constitute a basic prerequisite before deriving a specific control procedure. This step is detailed in Sec.~\ref{secformu}.

We now explore the optimal control of this system. We first use the circular geometry of the dynamics to simplify the corresponding equations. We introduce the angle $\theta$ such that $x_1=\cos\theta$ and $y_2=\sin\theta$, with $\theta(0)=1$. We arrive at
$$
\dot{\theta}=-u(t), 
$$
where the target state is here defined as $\theta_{\textrm{fi}}=\pm\frac{\pi}{2}$.
By symmetry, we can fix without loss of generality $\theta_{\textrm{fi}}=-\frac{\pi}{2}$.
Many control solutions $u$ exist to reach this state and a specific protocol can be selected by minimizing at the same time a functional of the state of the system and of the control, called a cost. Here, an example is given by the control time. To summarize in this example, the goal of the optimal control procedure is then to find the control $u$ steering the system to the target state in minimum time. Consider first constant controls $u(t)=u_0$ with $u_0\in\mathbb{R}$. The duration of the process is thus $T=\frac{\pi}{2u_0}$. This solution reveals a key problem in optimal control which corresponds to the existence of a minimum. In this example, arbitrarily fast controls can be achieved by considering larger and larger amplitudes $u_0$ and an optimal trajectory minimizing the transfer time does not exist. The analysis of the existence of optimal solutions which is a building block of any rigorous description of an optimal control problem is discussed in Sec.~\ref{s-existence}. It can be shown with the results presented in Sec.~\ref{s-existence} that an optimal solution exists if the set of available controls is restricted to a bounded interval, e.g., $u(t)\in [-u_m,u_m]$ where $u_m$ is the maximum amplitude. In this case, the optimal pulse is the control of maximum amplitude, the minimum time being equal to $\pi/{2u_m}$.

In order to illustrate the method of solving an optimal control problem, we consider the same transfer but in a fixed time $T$, the goal being to minimize the energy associated with the control, i.e., the term $\int_0^{T}\frac{u(t)^2}{2}dt$. There is no additional constraint on the control and we have $u(t)\in\mathbb{R}$. The target state is reached if $\int_0^{T}u(t)dt
=\frac{\pi}{2}$. Introducing the Lagrange multiplier $\lambda\in \mathbb{R}$, this constrained optimization problem can be transformed into the minimization of the functional
$$
J=\int_0^{T}\left(\frac{u(t)^2}{2}+\lambda u(t)
\right)dt-\lambda \frac{\pi}{2}.
$$
If we denote by $H$ the function $H=-\frac{u^2}{2}-\lambda u
$, the Euler-Lagrange principle leads to $\frac{\partial H}{\partial u}=0$, i.e., $u(t)=-\lambda
$. Using the constraint on the dynamics, we finally arrive at the optimal control $u(t)=\frac{\pi}{2T}$.

In this simple example, the optimal solution can be derived without the complete machinery of the Pontryagin Maximum Principle presented below. However, in the example we have introduced the main tools used in the PMP,
such as the Lagrange multiplier, the Pontryagin Hamiltonian $H$, and the maximization condition $\frac{\partial H}{\partial u}=0$. A few comments are in order here. The dynamical constraint is quite simple since the dimension of the state space is the same as the number of controls, the dynamics can be exactly integrated and the set of controls satisfying the constraint $\int_0^T u(t)dt=\frac\pi2$ is regular. This is not the case for a general nonlinear control system for  which (1) the Lagrange multiplier (which usually is  not constant but a function of time) is not easily found, (2) abnormal Lagrange multipliers appear if the set of controls satisfying the constraint is not regular.
We observe that $H$ can be rewritten as 
$$
H=\lambda\dot{\theta}-\frac{u^2}{2},
$$
which corresponds to the general formulation of the Pontryagin Hamiltonian in the normal case. The maximization condition remains the same in a general setting if there is no constraint on the available control. These aspects are discussed in details in Sec.~\ref{s-first}.

\section{Formulation of { the control problem}}\label{secformu}

\noindent{\bf The dynamics}.
 A finite-dimensional control system is a dynamical system governed
by an equation of the  form
\begin{equation}
\dot\qq(t)= f(\qq(t),u(t)),
\label{control}
\end{equation}
where
 $\qq:I\to M$ represents the state of the system, $I$ is an interval in $\R$ and $M$ is a smooth manifold whose dimension is denoted by $n$~\cite{comp1}. We recall that a manifold is a space that locally (but possibly not globally) looks like $\R^n$. Manifolds appear naturally in quantum control to describe, for instance, the $(2N-1)$-dimensional sphere $S^{2N-1}$, which is the set of wave functions of a $N$-level quantum system. The control law is $u:I\to U\subset\R^m$ and $f$ is a smooth function
such that $f(\cdot,\bar u)$ is a vector field on $M$ for every $\bar u\in U$.
An example of set $U$ of possible values of $u(t)$ is given by $U=[-1,1]^m$, meaning that the size of each of the coordinates of $u$ is at most one.
The set $U$ can  be the entire $\R^m$ if there is no control constraint.

To be sure that Eq.~\eqref{control} is well-posed from a mathematical viewpoint, we consider the case in which $I=[0,T]$ for some $T>0$ and
$u$ belongs to a space of regular enough functions $\mathcal{U}$ called the class of {\em admissible controls} (see \cite{comp3} for a precise definition).
Piecewise continuous controls
form a subset of admissible controls, and in experimental implementations in quantum control they are the only control laws that can be  reasonably applied.  However,
the class of piecewise constant controls is not suited to prove existence of optimal controls \cite{comp4}.

%

Given an admissible control $u(\cdot)$ and an initial condition $\qq(0)=\qqin\in M$, there exists a unique solution $\qq(\cdot)$ of Eq.~\eqref{control}, defined at least for small times~\cite{comp2}. A continuous curve $\qq(\cdot)$ for which there exists an admissible control  $u(\cdot)$ such that Eq.~\eqref{control} is satisfied is said to be an  {\em admissible trajectory}.

Let us present some
typical situations encountered in quantum control.  

Consider the time evolution of the wave function of a closed $N$-level quantum system. In this case, under the dipolar approximation~\cite{RMPsugny,lapert2008,ohtsuki2008}, the dynamics are governed by the Schr\"odinger equation (in units where $\hbar=1$)
\begin{equation}\label{eqwave}
i\dot\psi(t)=\left(H_0+\sum_{j=1}^m u_j(t)H_j\right)\psi(t),
\end{equation}
where $\psi$, the \emph{wave function}, belongs to the unit sphere in $\C^N$ and $H_0,\ldots,H_m$ are $N\times N$ Hermitian matrices.  The control parameters  $u_j(t)\in\R$ are the components of the control $u(\cdot)$. This control problem has the form \eqref{control} with $n=2N-1$,
$M=S^{2N-1}\subset\C^N$, $\qq=\psi$,
and $f(\psi,u)=-i (H_0+\sum_j^m u_j H_j)\psi$. Note that the uncontrolled part corresponding to the $H_0$- term is called the drift.  The solution of the Schr\"odinger equation can also be expressed in terms of the unitary operator ${\bf U}(t,t_0)$,
which connects the wave function at time $t_0$ to its value at $t$: $\psi(t)={\bf U}(t,t_0)\psi(t_0)$. The \emph{propagator} ${\bf U}(t,t_0)$ also satisfies the Schr\"odinger equation
\begin{equation}\label{eqprop}
i\dot {\bf U}(t,t_0)=\left(H_0+\sum_{j=1}^m u_j(t)H_j\right){\bf U}(t,t_0),
\end{equation}
with  initial condition ${\bf U}(t_0,t_0)=\mathbb{I}_N$. In quantum computing, the control problem is generally defined with respect to the propagator ${\bf U}$. Equation~\eqref{eqprop} has the form \eqref{control} with $M=U(N)\subset \C^{N\times N}$ and $\qq={\bf U}$.

The wave function formalism is well adapted to describe pure states of isolated quantum systems, but when one lacks information about the system the correct
formalism is the one of mixed-state quantum systems. The state of the system is then described by a density operator $\rho$, which is a $N\times N$ positive semi-definite Hermitian matrix of unit trace. For a closed quantum system, the density operator is a solution  of the von Neumann equation
$$
\dot \rho(t)=-i[H,\rho(t)],
$$
with $H=H_0+\sum_{j=1}^m u_j(t)H_j$. For an open $N$-level quantum system interacting with its environment, the dynamics of $\rho$ are governed in some cases~\cite{giorgi,breuerbook} by the following first-order differential equation, called the Kossakowski--Lindblad equation~\cite{gorini,lindblad}:
\begin{equation}\label{dissipative}
\dot\rho(t)=-i[H,\rho(t)]+\mathcal{L}[\rho(t)].
\end{equation}

This equation differs from the von Neumann one in that a dissipation operator $\mathcal{L}$ acting on the set of density operators has been added. This
linear operator which describes the interaction with the environment cannot be chosen arbitrarily. Its expression can be derived from physical arguments based on a Markovian regime and a small coupling with the environment~\cite{alickibook,breuerbook}. From a mathematical point of view, the problem of finding dynamical generators for open systems that ensure complete positivity of the dynamical evolution was solved in finite- and infinite-dimensional Hilbert spaces~\cite{gorini,lindblad}.  The operator $\mathcal{L}$ is a linear operator acting on the space of density matrices that can be expressed for a $N$-level quantum system as
\begin{align*}
\mathcal{L}[\rho(t)]&=\frac{1}{2}
\sum_{k,k'=1}^{N^2-1}a_{kk'}([V_k\rho(t),V_{k'}^\dagger]+[V_k,\rho(t)V_{k'}^\dagger]),
\end{align*}
where the matrices $V_1,\dots,V_{N^2-1}$ are trace-zero and orthonormal. The linear mapping $\mathcal{L}$ is completely positive if and only if the matrix $a=(a_{kk'})_{k,k'=1}^{N^2-1}$ is positive~\cite{schirmer2004,Schirmer_2004}.
The density operator $\rho$ can be represented as a vector $\vec{\rho}$ by stacking its columns. The corresponding time evolution is generated by superoperators in the Schr\"odinger-like form
\begin{equation}\label{eqliouville}
i\dot{ \vec{\rho}}=\textbf{H}\vec{\rho}.
\end{equation}
Equation~\eqref{eqliouville} has the form \eqref{control} with $M=\mathbb{B}^{N^2-1}\subset \mathbb{R}^{N^2-1}$, $\qq=\vec{\rho}$, and $f:\vec{\rho}\mapsto \textbf{H}\vec{\rho}$. Here $\mathbb{B}^{N^2-1}$ denotes the ball of radius 1 in $\mathbb{R}^{N^2-1}$~\cite{comp5}.

\noindent{\bf The initial and final states}.
{When considering a quantum control problem, the goal in most situations is not to bring the system from an initial state $\qqin$ to a final state $\qqfin$, but rather to reach at time $T$ a smooth \emph{submanifold} $\target$ of $M$ (see \cite{submanifold} for a precise definition), called {\em target}}:
\begin{equation}
\qq(0)=\qqin,\qquad \qq(T)\in\target.
\label{infin}
\end{equation}
{This issue arises, for instance, in the population transfer from a state $\psi_{\textrm{in}}$
 to an eigenstate $\psi_{\textrm{fi}}$ of the field-free Hamiltonian $H_0$. In this case, since the phase of the final state is not physically relevant, $\target$ is characterized by  $\{e^{i\theta} \psi_{\textrm{fi}}\mid \theta\in[0,2\pi]\}$.
It can also happen that the initial condition $\qq(0)=\qqin$ is generalized to $\qq(0)\in{\cal S}$, where ${\cal S}$ is a smooth submanifold of $M$. However, for the sake of presentation, we will not treat this case here, the changes to be made to the method being straightforward. Finally, note that the time $T$ can be fixed or free, as,  for instance, in a time-minimum control problem.}\\
\begin{figure}
\begin{center}
\begin{flushleft}
\includegraphics[scale=0.6]{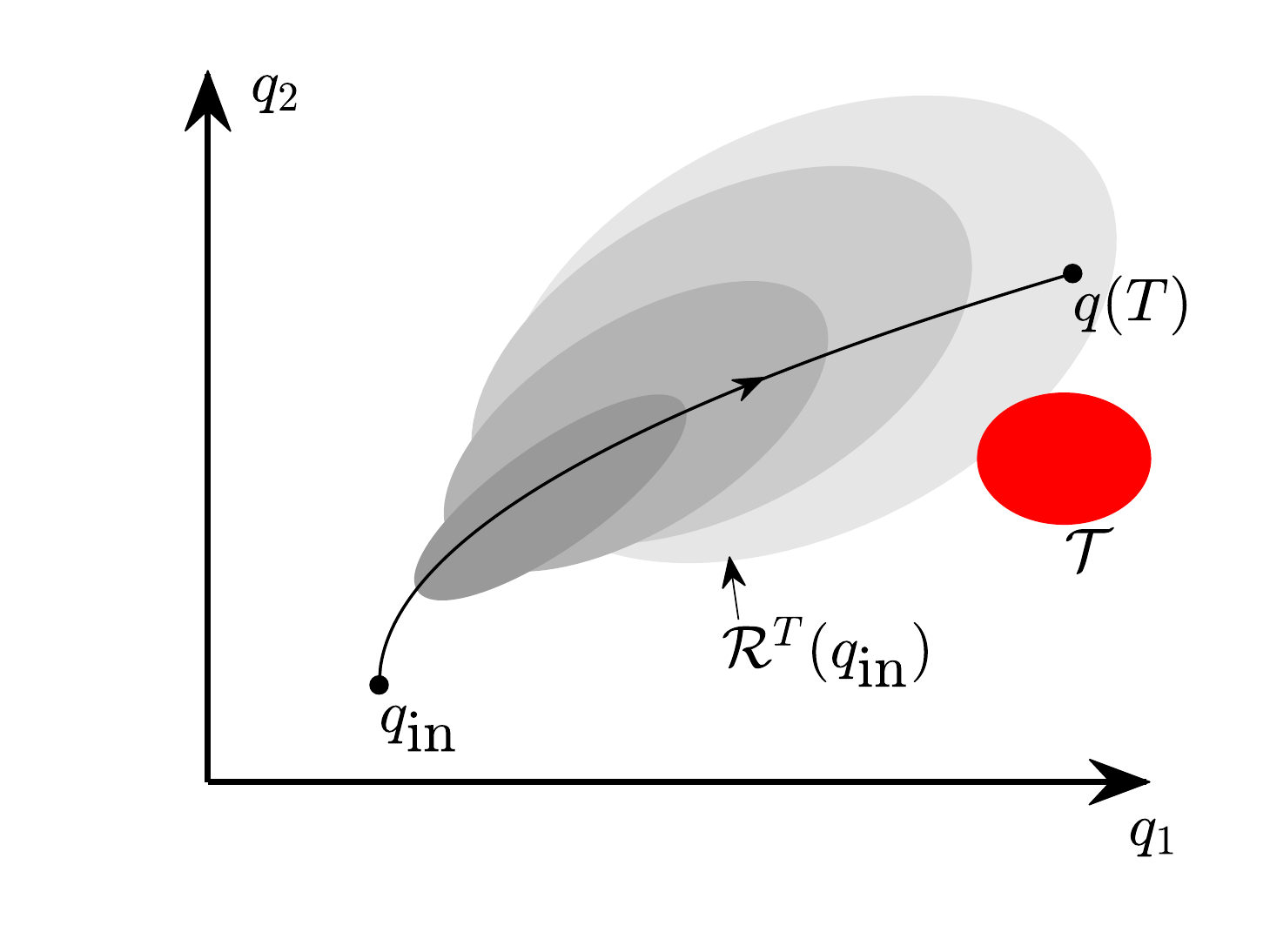}
\end{flushleft}
\end{center}
\caption{The reachable set $\mathcal{R}^T(q_{\textrm{in}})$ from $q_{\textrm{in}}$ at different times $T$ (the shades of gray indicate an increase of the control time) and the target $\mathcal{T}$ (red area). The manifold $M$ is $\mathbb{R}^2$ and the coordinates of $q$ are $(q_1,q_2)$. The intersection between the two sets $\mathcal{R}^T(q_{\textrm{in}})$ and $\mathcal{T}$ is non-empty for a long enough time $T$. Note that the initial point of the dynamics does not belong to the reachable set at time $T$. This is due to a specific choice of the dynamical system with a drift.\label{f0}}
\end{figure}

\noindent {\bf The optimal control problem}.   Two different optimal control approaches can be used to steer the system from $\qqin$ to a target $\target$.\\
$\bullet$ Approach A: Prove that the target $\target$ is reachable from $\qqin$ (in time $T$ if the final time is fixed or in any time otherwise) and then find the best possible control realizing the transfer. This approach requires to solve the preliminary step of controllability. Essentially, we need to show that:
\begin{align}
& \target\cap {\cal R}(\qqin)\not=\emptyset \mbox{ if $T$ is free where}\nn\\
&{\cal R}(\qqin):=\{\bar\qq\in M\mid{} \exists~T \mbox{ and}\nn\\
& \mbox{an admissible trajectory } q:[0,T]\to M \nn\\
&\mbox{ such that }\qq(0)=\qqin, ~\qq(T)=\bar\qq\}\nn
\end{align}
or that
\begin{align}
& \target\cap {\cal R}^T(\qqin)\not=\emptyset \mbox{ if $T$ is fixed where}\nn\\
& {\cal R}^T(\qqin):=\{\bar\qq\in M\mid{} \exists\mbox{ an admissible trajectory }\nn\\
&q:[0,T]\to M \mbox{ such that }\qq(0)=\qqin, ~\qq(T)=\bar\qq\},\nn
\end{align}
 and then solve the minimization problem
\begin{equation}\label{pbA}
\int_0^{T} f^0(\qq(t),u(t))\,dt \longrightarrow \min,
\end{equation}
where $f^0:M\times U\to \R$ is a smooth function,
 which in many quantum control applications  depends only on the control. An example is given by the functional $\int_0^{T}
 u^2(t)
 dt$ which represents the energy used in the control process. The control time $T$ is fixed or free. This integral is generally called the \emph{cost functional}. A schematic illustration of the reachable set $\mathcal{R}^T(\qqin)$ and the target $\mathcal{T}$ is given in Fig.~\ref{f0}.

The test of controllability is sometimes easy (as, for instance, for low-dimen\-sional closed quantum systems~\cite{albertini-dalessandro,schirmer:2001}) and sometimes 
more difficult. We recall that a closed quantum system is controllable if the matrix Lie algebra generated by the matrices  $H_0,\dots,H_m$
is $\mathrm{su}(N)$.
For general systems, a
useful
sufficient condition for controllability is described in Appendix~\ref{testcontrollability}. When the test of controllability can be performed, this approach is to be preferred since it permits to reach exactly the final state.

\noindent $\bullet$  Approach B: Find a control that brings the system as close as possible to the target, while minimizing the cost. This approach is used for systems for which the controllability step cannot be easily verified. In this case, the initial point is fixed and the final point is free, but the cost contains a term (denoted $d(\cdot,\cdot)$ in the next formula) depending on the distance between the final state of the dynamics and the target:
\begin{equation}\label{pbB}
\int_0^{T} f^0(\qq(t),u(t))\,dt+d(\target,\qq(T))\rightarrow \min,
\end{equation}
where $T$ is fixed or free.

\begin{example}
An example is given by open quantum systems governed by the Kossakowski--Lindblad equation, for which the characterization of the reachable set is quite involved~\cite{altafini2003,dirr2009}. If we denote by $\rho_{\textrm{fi}}$ the target state, a cost functional to minimize penalizing the energy of the control and the distance to the target can be
$$
 \int_0^T\sum_{j=1}^m\frac{u_j(t)^2}{2}dt+\|\rho(T)-\rho_{\textrm{fi}}\|^2,
$$
where $\|\cdot\|$ is the norm  corresponding to the scalar product of density matrices $\langle \rho_1|\rho_2\rangle =\textrm{Tr}[\rho_1^\dagger\rho_2]$.
\end{example}

Optimization problem in these two approaches should be of course considered together with the dynamics~\eqref{control} and the initial and final conditions.
They are summarized in Tab.~\ref{tabsumm}.

\begin{table}[tb]
\caption{Summary of the different optimal control approaches.\label{tabsumm}}
\begin{tabular}{|c|c|}
\hline
{\bf Approach A}&$\dot\qq(t)=f(\qq(t),u(t))$\\
When controllability can be verified, i.e., one can prove that:&$\qq(0)=\qqin,~~\qq(T)\in\target$\\
 $\target\cap{\cal R}(\qqin)\not=\emptyset$ if $T$ is free or&$\int_0^T f^0(\qq(t),u(t))\,dt\to \min$\\
$\target\cap{\cal R}^T(\qqin)\not=\emptyset$ if $T$ is fixed &$T$ fixed or free\\
\hline\hline
{\bf Approach B}&$\dot\qq(t)=f(\qq(t),u(t))$\\
When controllability cannot be verified&$\qq(0)=\qqin,~~\qq(T)$ free\\
&$\int_0^T f^0(\qq(t),u(t))\,dt+d(\target,\qq(T))\to\min$\\
&$T$ fixed or free\\
\hline
\end{tabular}
\end{table}

\section{The different steps to solve an optimal control problem}\label{secsteps}
The steps to determine a solution to the minimization problems~\eqref{pbA} and \eqref{pbB} are similar to finding the minimum of a smooth function $f^0:\R\to\R$.

\begin{itemize}
\item[0.] \textbf{Find conditions which guarantee the existence of solutions.} We recall that among smooth functions $f^0:\R\to\R$, it is easy to find examples not admitting a minimum (e.g., the function $x\mapsto e^{-x}$ and the function $x\mapsto x$ do not have minima).  This step is crucial. If it is skipped, first-order conditions may give a wrong candidate for optimality (see below for details) and numerical optimization schemes may either not converge or converge towards a solution which is not a minimum.  For optimal control problems, there exist several existence tests, but they are not always applicable or easy to use. In Sec.~\ref{s-existence}, we present the Filippov test.

\item[1.] \textbf{Apply first-order necessary conditions.} For a smooth function $f^0:\R\to\R$, this means  that $\textrm{if}~\bar x~\textrm{is a minimum then}~\frac{d}{dx}f^0(\bar x)=0.$ This condition gives candidates for minima, i.e., identifies local minima, local maxima, and saddles. Note that if one does not verify a priori existence of minima, first-order conditions could give wrong candidates. Think for instance to the function $x\mapsto (x^2+1/2)e^{-x^2}$. This function has a single local minimum, obtained at  $x=0$, whose value is $1/2$, which is well identified by first-order conditions. However its infimum is zero (for $x\to\pm\infty$, the function tends to zero). For optimal control problems, first-order necessary conditions should be given in an infinite-dimensional space (a space of curves) and they are expressed by the PMP, which is presented in Sec.~\ref{ss-PMP}. In Approach A, note that the condition that the system reaches exactly the target is a {\em constraint} leading to the appearance of {\em Lagrange multipliers (normal and abnormal)}. This point is discussed in details in Sec.~\ref{ss-lagrange}.

\item[2.] \textbf{Apply second-order conditions.}  For instance, for a smooth function $f^0:\R\to\R$, among the points at which we have $\frac{d}{dx}f^0(\bar x)=0$, a necessary condition to have a minimum is $\frac{d^2}{dx^2}f^0(\bar x)\geq0$. This step is generally used to reduce further the candidates for optimality. For optimal control problems, there are several second-order conditions, such as higher-order Pontryagin Maximum Principles or Legendre--Clebsch conditions (see for instance \cite{agrachev-book,schaettler-book,boscain-book}). In some cases, this step is difficult and it could be more convenient to go directly to the next one.

\item[3.]  \textbf{Selection of the best solution among all candidates.} Among the set of candidates for optimality identified in step 1 and (possibly) further reduced in step 2, one should select the best one. This step is often done by hand if the previous steps have identified a finite number of candidates for optimality.  For optimal control problems, one often ends up with infinitely many candidates for optimality and this step is generally very difficult.

\end{itemize}
 There are of course specific examples for which the solution is particularly simple. This is the case of convex problems, for which
 only first-order conditions should be applied, since the existence step is automatic and first-order conditions are both necessary and sufficient for optimality.  This situation is however rare in quantum control
and we will not discuss it further.

\section{Existence of solutions for Optimal Control Problem: the Filippov test}\label{s-existence}

The existence theory for optimal control is difficult and, unfortunately, there is no general procedure that can be applied in any situation. In this section, we present the most important technique, the Filippov test that allows to tackle several types of problems. In order to keep this paragraph as accessible as possible, we present below only the main ideas and some propositions derived from the Filippov test. These results can be directly applied to quantum systems. A complete statement of the Filippov test is provided in Appendix~\ref{filippovtheorem}. We emphasize that it is fundamental to verify the existence of optimal controls before applying first-order conditions (i.e., the PMP). Otherwise, as discussed in the finite-dimensional case, it may occur that the PMP has solutions, but none of them is optimal.

Let us consider the problem in Approach A with $T$ fixed.

\noindent {\bf Problem P1}
\begin{align}
&\dot\qq(t)=f(\qq(t),u(t)),\nn\\
&\qq(0)=\qqin,~~\qq(T)\in\target,\nn\\
&\int_0^T f^0(\qq(t),u(t))\,dt\to \min,\nn\\
&T>0~\mbox{\rm fixed}.\nn
\end{align}
{\em Here $q:[0,T]\to M$, where $M$ is a smooth $n$-dimensional manifold, $f,f^0$ are smooth functions of their arguments, $u\in \mathcal{U}$, $U\subset\R^m$, and $\target$ is a smooth submanifold of $M$.
}

In order to tackle the existence problem, we define a new variable $\qq^0$ obtained as the value of the cost during the time-evolution, that is,
$$\qq^0(t)=\int_0^t f^0(\qq(s),u(s))\,ds,$$
and we denote  $\hat \qq=(\qq^0,\qq)$. The dynamics of the new state $\hat q$ in $\R\times M$ are given by
\begin{align}
\dot{\hat \qq}(t)&=\left(\ba{c}\dot \qq^0(t)\\\dot \qq(t)\ea \right)=\left(\ba{c}f^0(\qq(t),u(t))\\ f(\qq(t),u(t))\ea \right)\nn \\
&=: \hat f(\qq(t),u(t)),\nonumber\\
&\hat{\qq}(0)=(0,\qqin),~~~ \hat{\qq}(T)\in {\bf \R}\times\target. \nonumber
\end{align}
This control system is called the {\it augmented system}. The minimization problem in integral form, $\min\int_0^Tf^0(q(t),u(t))\,dt$, becomes a problem of minimization of one of the coordinates at the final time, i.e., $\min\qq^0(T)$.

We denote by $\hat{\cal R}^T(0,\qqin)$ the reachable set at time $T$ starting from  $(0,\qqin)$ for the augmented system. The key observation on which optimal control is based is expressed by the  following proposition.
\begin{prop}
If $q(\cdot)$ is an optimal trajectory for problem {\bf (P1)}, then $\hat q(T)\in\partial \hat{\cal R}^T(0,\qqin).$
\end{prop}

\begin{proof}
By contradiction, if $\hat \qq(T)=(\qq^0(T),\qq(T))\in \mbox{int}\hat{\cal R}^T(0,\qqin)$ then there exists a trajectory reaching a point $(\al,\qq(T))$ with $\al<\qq^0(T)$, i.e., arriving at the same point in $M$, but with a smaller cost. See Fig.~\ref{f1}.
\end{proof}
\begin{figure}
\begin{center}
\includegraphics[scale=0.5]{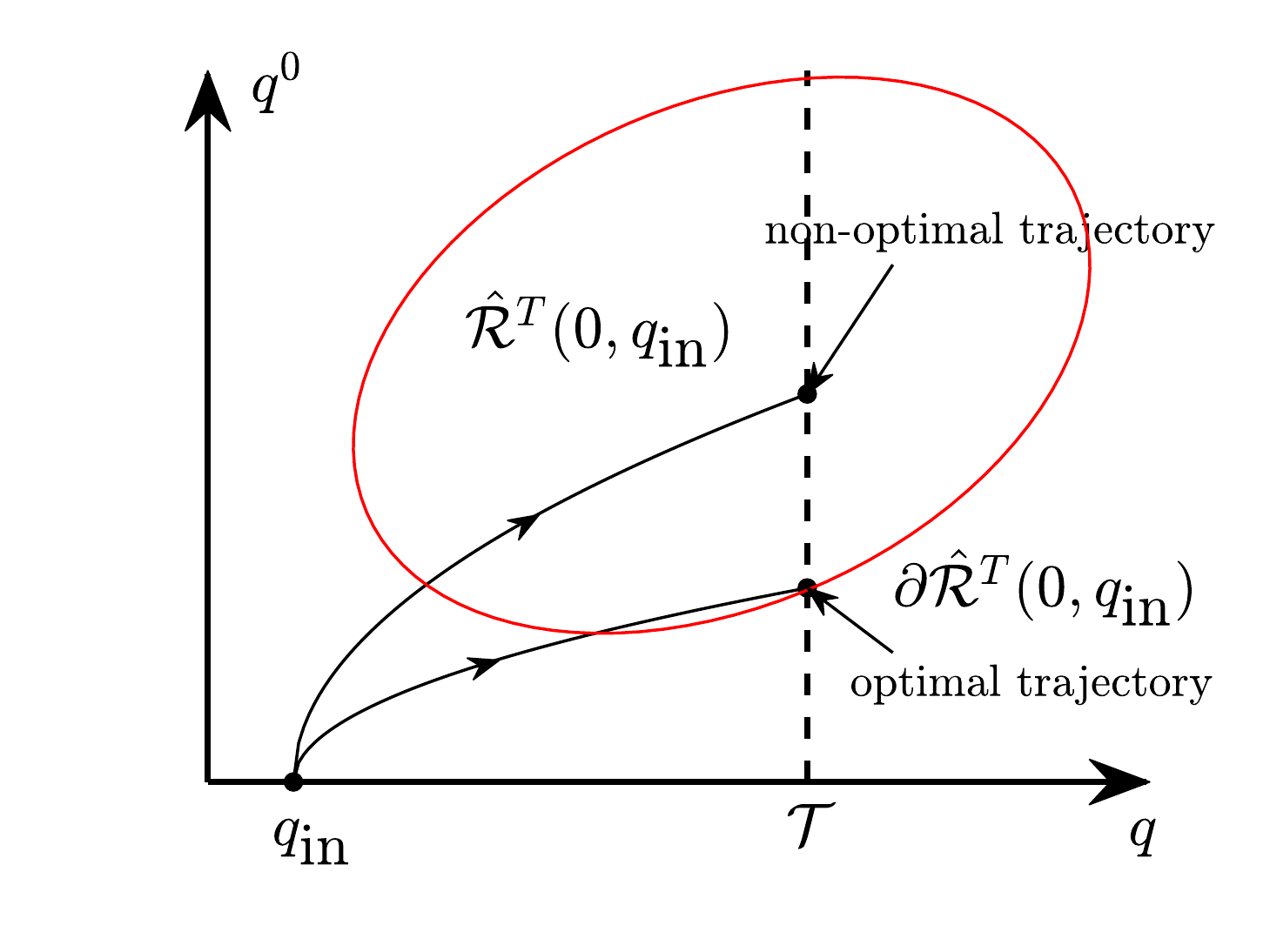}
\end{center}
\caption{The reachable set of the augmented system (area delimited by the red curve). The target $\mathcal{T}$ is represented by the vertical dashed line.\label{f1}}
\end{figure}

When $\hat{\cal R}^T(0,\qqin)\cap(\R\times\target)$ is nonempty and compact, an optimal trajectory for problem {\bf (P1)} exists~\cite{comp6}. We deduce the following property.

\bp
If $\hat{\cal R}^T(0,\qqin)$ is compact, $\target$ is closed, and
$\hat{\cal R}^T(0,\qqin)\cap(\R\times\target)$ is nonempty,
then there exists a solution to   problem {\bf (P1)}.
\ep

Hence the compactness of $\hat{\cal R}^T(0,\qqin)$ is a key point. A similar reasoning allows to relate the compactness of the reachable set within time $T$ and the existence of solutions for the minimum time problem. A sufficient condition for compactness of the reachable set is given by Filippov's theorem, which is stated in Appendix~\ref{filippovtheorem}.
A consequence of Filippov's theorem is the following (see also Propositions~\ref{p-teschio} and
\ref{p-tempomin}
in Appendix~\ref{filippovtheorem}).
%

\begin{prop}\label{p-teschio-quantico}
Consider the Schr\"odinger equation 
\begin{equation}\label{eqwave-tt}
i\dot\psi(t)=\left(H_0+\sum_{j=1}^m u_j(t)H_j\right)\psi(t),
\end{equation}
where $\psi(t)$ evolves in the unit sphere $S^{2N-1}$ of $\C^N$, the matrices $H_0,\ldots,H_m$ are $N\times N$ Hermitian, and $u(t)\in U$. Let $\psi_0$
be the initial condition for $\psi(\cdot)$ and the final target $\mathcal{T}$ be closed.
Assume that the set $U$ is convex and compact.
Then each of the following two optimal control problems admits a solution:
\begin{enumerate}
\item $T$ is fixed, ${\cal R}^T(\qqin)$ intersects $\mathcal{T}$, and the cost function $f^0:S^{2N-1}\times U\to \R$ is convex;
\item  minimum time problem when
${\cal R}(\qqin)$ intersects $\mathcal{T}$.
\end{enumerate}
\end{prop}

\begin{example}
The compactness of $U$ is a key assumption to ensure the existence of an optimal solution. We come back to the example of Sec.~\ref{introqual} and we consider the control of a two-level quantum system whose dynamics is governed by the Hamiltonian $H=u\left(\begin{smallmatrix}0&1\\1&0\end{smallmatrix}\right)$ with $u$ real-valued
~\cite{boscain-q1}. The goal of the control is to steer the wave function from $\psi_{\textrm{in}}=(1,0)$ to the 
target ${\cal T}=\{(0,e^{i\theta})\mid \theta\in \R\}$ in minimum time. If $U=\mathbb{R}$ then there is no optimal solution since the target state can be reached in an arbitrary small time, while the control cannot be realized in a zero time. When $U=(-1,1)$, the transfer can be achieved in a time larger than $\pi/2$, but not exactly in time $T=\pi/2$. It is only in the case where $U$ is compact, e.g. $U=[-1,1]$, that an optimal solution exists. We find for the constraint $|u(t)|\leq 1$ a pulse which allows to bring the population from one level to another in a time $T=\pi/2$.
\end{example}

We show in Sec.~\ref{secthree} and Sec.~\ref{sectwo} how to use these results in two standard quantum control examples.

\section{First-order conditions}\label{s-first}

For a smooth real-valued function of one variable $f^0:\R\to\R$, first-order optimality conditions are obtained from the observation that, at  points where $\frac{df^0}{dx}\neq0$, the function $f^0$, which  is well approximated by its first-order Taylor series, cannot be optimal since it behaves  locally as a non-constant affine function. In this way, one obtains the necessary condition: {\em If $\bar x$ is minimal for $f^0$ then $\frac{df^0}{dx}(\bar x)=0$}. First-order conditions in optimal control are derived in the same way. We have to require that for a small control variation, there is no cost variation at first order.

More precisely, if $J(u(\cdot))$ is the value of the cost for a reference admissible control $u(\cdot)$ (for instance $J(u(\cdot))=\int_0^T f^0(\qq(t),u(t))\,dt$ in Approach A or
$J(u(\cdot))=\int_0^T f^0(\qq(t),u(t))\,dt+d(\target,\qq(T))$ in Approach B),
and $v(\cdot)$ is another admissible control, one would like to consider a condition of the form
\begin{equation}
\left.\frac{\partial J(u(\cdot)+h v(\cdot))}{\partial h}\right|_{h=0}=0.
\label{ccc}
\end{equation}
But difficulties may arise for the following reasons.

We work in an infinite-dimensional space (the space of controls) and hence condition \eqref{ccc} should be required
for infinitely many $v(\cdot)$.
It may very well happen that if $u(\cdot)$ and $v(\cdot)$ are admissible controls then  $u(\cdot)+h v(\cdot)$ is not admissible for every  $h$ close to $0$. Think, for instance, to the case in which $m=1$ and $U=[a,b]$. If $u(t)\equiv b$ is the reference control, then $u(t)+hv(t)$ is not admissible for any non-zero perturbation $v(\cdot)$ when $hv(t)$ is strictly positive. Hence, one should be very careful in choosing the admissible variations in order to fulfill the control restrictions. In Approach  A, one should restrict only to control variations for which the corresponding trajectory reaches the target. More precisely, if $\tilde\qq(\cdot)$ is the trajectory corresponding to the control $\tilde u(\cdot):=u(\cdot)+h v(\cdot)$, one should add the condition
\begin{equation}
\tilde \qq(T)\in\target,
\label{constraint}
\end{equation}
 with $T$ either free or constrained to be the fixed final time depending on the problem under study.
Condition~\eqref{constraint} should be considered as a {\em constraint} for the minimization problem, which results in the use of {\em  Lagrange multipliers (normal and abnormal)}.

The occurrence of Lagrange multipliers in optimal control
is therefore not due to
the fact that the optimization takes place in an infinite-dimensional space, but is rather a general feature of
constrained minimization problems, as explained in Sec.~\ref{ss-lagrange}.

\subsection{Why Lagrange multipliers appear in constrained optimization problems}\label{ss-lagrange}

 We first recall how to find the minimum of a function of $n$ variables $f^0(x)$, where $x=(x_1,\dots,x_n)$, under the constraint $f(x)=0$,  with the method of Lagrange multipliers. Here $f^0$ and $f$ are two smooth functions $\R^n\to\R$.
We have two cases.

\noindent $\bullet$ If $\bar x$ is a point such that $f(\bar x)=0$ with $\nabla f(\bar x)\neq0$, then the implicit function theorem guarantees that $\{x\mid f(x)=0\}$ is a smooth hypersurface in a neighborhood of  $\bar x$.  In this case, a necessary condition for $f^0$ to have a minimum at $\bar x$ is that the level set of $f^0$ (i.e., the set
 on which $f^0$ takes
a constant value) is
not transversal to the set $\{x\mid f(x)=0\}$ at $\bar x$. See Figure \ref{f2}.

More precisely, this means that
\begin{equation}
\exists \lambda\in\R\mbox{ such that  }  \nabla f^0(\bar x)=\lambda \nabla f(\bar x).
\label{e-lagrange-1}
\end{equation}
This statement can be proved by assuming, for instance,  that $\partial_{x_n}f(\bar x)\neq0$. The set $\{x\mid f(x)=0\}$ can then be expressed locally around $\bar x$ as $x_n=g(x_1,\dots,x_{n-1})$. The requirement that
\begin{align*}
& \partial_{x_i}f(x_1,\dots,x_{n-1},g(x_1,\dots,x_{n-1}))\equiv 0,\\
& i=1,\dots,n-1,\\
& \partial_{x_i}f^0(x_1,\dots,x_{n-1},g(x_1,\dots,x_{n-1}))|_{x=\bar x}=0,\\
& i=1,\dots,n-1,
\end{align*}
provides immediately condition~\eqref{e-lagrange-1} with $\lambda=\frac{\partial_{x_n} f^0(\bar x)}{\partial_{x_n} f(\bar x)}$.

\begin{figure}
\begin{center}
\includegraphics[scale=0.5]{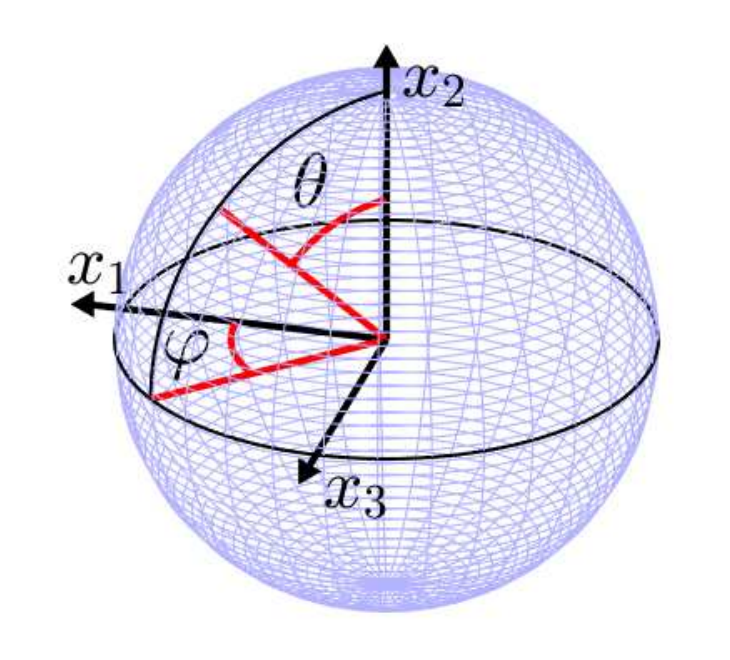}

\end{center}
\caption{ Intersection of the set $f(x)=0$ (solid line) with the level set of $f^0$ (dashed line). The two gradients $\nabla f(x)$ and $\nabla f^0(x)$ are parallel  at
$x=\bar{x}$.\label{f2}}
\end{figure}
Notice that $\lambda$ could be equal to zero. This case corresponds to the situation in which $f^0$ has a critical point at $\bar x$ even in the absence of the constraint.

\noindent $\bullet$ If $\bar x$ is a point such that $f(\bar x)=0$ with $\nabla f(\bar x)=0$ then the set $\{x\mid f(x)=0\}$ could be very complicated in a neighborhood of $\bar x$ (typical examples are a single point, two crossing curves, \dots\ but it could be any closed set).
In general the value of $f^0$ at these points  cannot be compared with neighboring points by requiring that a certain derivative is
zero (think for instance to the case in which $\{x\mid f(x)=0\}$ is an isolated point). However, they are candidates to optimality.
As an illustrative example, consider the case where $n=2$, $f^0(x_1,x_2)=x_1^2+(x_2-1/4)^2$, and $f(x_1,x_2)=(x_1^2+x_2^2)(x_1^2+x_2^2-1)$. The point $\bar{x}=(0,0)$ is an isolated point for which $f(\bar{x})=0$ and $\nabla f(\bar{x})=0$.\\

These results can be rewritten in the following form.
\bt[Lagrange multiplier rule in $\R^n$]\label{t-LM}
Let $f^0$ and $f$ be two smooth functions from $\R^n$ to $\R$. If $f^0$ has a minimum at $\bar x$ on the set $\{x\mid f(x)=0\}$, then there exists $(\bar \lambda,\bar \lambda_0)\in\R^2\setminus\{(0,0)\}$
such that, setting $\Lambda(x,\lambda,\lambda_0)=\lambda f(x)+\lambda_0 f^0(x)$, we have
 \begin{equation}
 \nabla_{x} \Lambda(\bar x,\bar \lambda,\bar \lambda_0)=0,~\nabla_{\lambda} \Lambda(\bar x,\bar \lambda,\bar \lambda_0)=0.
 \label{e-lag2}
 \end{equation}
\et

To show that this statement is equivalent to what we just discussed, we observe that
the second equality in \eqref{e-lag2} gives the constraint $f(\bar x)=0$. For the first equation, we have two cases. If $\bar \lambda_0\neq0$ then we can normalize $\bar \lambda_0=-1$ and we get $\bar \lambda \nabla_x f(\bar x)-\nabla_x f^0(\bar x)=0$, i.e., Eq.~\eqref{e-lagrange-1} with the change of notation $\lambda\to\bar \lambda$. If $\bar \lambda_0=0$ then $\bar \lambda\neq0$ and we get $ \nabla_x f(\bar x)=0$, that is, the second case studied above.

The quantities $\bar \lambda$ and $\bar \lambda_0$ are respectively called {\em Lagrange multiplier} and {\em abnormal Lagrange multiplier}. If $(\bar x ,\bar\lambda_0,\bar \lambda)$ is a solution of Eq.~\eqref{e-lag2} with $\bar \lambda_0\neq0$ (resp., $\bar \lambda_0=0$) then $\bar x$ is called a {\em normal extremal} (resp., {\em abnormal extremal}). An abnormal extremal is a candidate for optimality
 and occurs, in particular, when
we cannot guarantee (at first order) that the set $\{x\mid f(x)=0\}$ is a smooth curve. Abnormal extremals are candidates for optimality regardless of cost $f^0$.
Note that if $\bar x$ is such that $\nabla_xf(\bar x)=0$ and $\nabla_xf^0(\bar x)=0$ then $\bar x$ is both normal and abnormal. This is the case in which $\bar x$ satisfies the first-order condition for optimality even without the constraint, but we cannot guarantee that the constraint is a smooth curve.

In the (infinite-dimensional) case  of an optimal control problem, normal and abnormal Lagrange multipliers appear in a very similar way.

\subsection{Statement of the Pontryagin Maximum Principle}\label{ss-PMP}
In this section,  we state the first-order necessary conditions for optimal control problems, namely the PMP.  The basic idea is to define a new object (the pre-Hamiltonian, see Eq.~\eqref{HAMHAMHAM} below) which allows
 to formulate the Lagrange multiplier
conditions in a simple and direct way.

The theorem is stated in a  more general form that unifies and slightly generalizes the optimal control problems of Approaches A and B. In particular, we add to the  cost $\int_0^Tf^0(\qq(t),u(t))dt$ a general terminal cost $\phi(\qq(T))$. In Approach~A,  we have $\phi=0$, while in Approach~B, $\phi$ represents
the distance from $\qq(T)$ to the target $\target$. We allow the target $\target$ to coincide with $M$. This
corresponds to leaving the final point $\qq(T)$ free
in Approach B.

\bt\label{t-PMP}
Consider the optimal control problem
\bqn
&&\dot \qq(t)=f(\qq(t),u(t)),\nn\\
&&\qq(0)=\qqin,~~~\qq(T)\in\target,\nn\\
&& \int_0^Tf^0(\qq(t),u(t))~dt +\phi(\qq(T))\longrightarrow \min,\nn
\eqn
where
\bi
\iii $M$ is a smooth manifold of dimension $n$, $U\subset \R^m$,
\iii $\target$ is a (non-empty) smooth submanifold of $M$. It can be reduced to a point (fixed terminal point) or coincide with $M$ (free terminal point),
\iii $f$, $f^0$ are smooth,
 \iii $u\in \mathcal{U}$,
\iii $\qq:[0,T]\to M$ is a continuous curve~\cite{comp2}.
\ei
Define
the function (called \emph{pre-Hamiltonian})
\begin{equation}
\HHH (q,p,u,p^0)=\langle p, f(q,u)\rangle+p^0f^0(q,u),\label{HAMHAMHAM}
\end{equation}
with
\[
(q,p,u,p^0)\in T^\ast M\times U\times\R.
\]
(see \cite{comp1} for a precise definition of $T^\ast M$).

If the pair $(\qq,u):[0,T]\to M\times U$ is
optimal, then there
exists a never vanishing continuous pair
$(p,p^0):[0,T]\ni t \mapsto
(p(t),p^0)\in
T^\ast_{\qq(t)}M\times \R$
where $p^0\leq0$ is  a constant and such that for almost every (a.e.) $t\in[0,T]$ we have
\begin{itemize}
\item[{\bf i)}]
$\dot \qq(t)=$ {\large $\frac{\partial \HHH
}{\partial p}$}$(\qq(t),p(t),u(t),p^0)$ (Hamiltonian equation for $\qq$);
\item[{\bf ii)}] $\dot p(t)=-${\large$\frac{\partial \HHH
}{\partial q}$}$(\qq(t),p(t),u(t),p^0)$ (Hamiltonian equation for $p$);
\item[{\bf iii)}]
the quantity
$\HHH_M(q(t),p(t),p^0):=\max_{v\in U}\HHH(q(t),p(t),v,p^0)$
is well-defined and
$$\HHH (\qq(t),p(t),u(t),p^0)=\HHH_M(\qq(t),p(t),p^0)$$
which corresponds to the maximization condition.
\end{itemize}
Moreover,
\begin{itemize}
\item[{\bf iv)}]
 there exists a constant $c\ge 0$ such that
$\HHH_M(\qq(t),p(t),p^0)=c
$ on  $[0,T]$,
with $c=0$ if
the final time
is free (value of the Hamiltonian);
\item[{\bf v)}] for every $v\in T_{\qq(T)}\target$, we have $\langle p(T),v\rangle=p^0\langle d\phi(\qq(T)),v\rangle$  (transversa\-lity condition),  where $d\phi$ is the differential of the function $\phi$.
\end{itemize}
\et

\noindent Some comments are in order.\\
\noindent $\bullet$
A proof of the PMP can be found, for instance, in \cite{agrachev-book,pontryaginbook}. An intuitive derivation based on the Lagrange multiplier rule is presented in Section~\ref{WPMP} in the case in which $T$ is fixed, the final point is free, and $U=\R^m$.

\noindent $\bullet$ The covector $p$ is called \emph{adjoint state} in the control theory literature~\cite{comp1}, while $p^0$ is the \emph{abnormal multiplier}. The quantities $p(\cdot)$ and $p^0$ play the role of Lagrange multipliers for the constrained optimization problem. We point out the similarity between the expressions of $\HHH$ and of $\Lambda$ in Th.~\ref{t-LM} (with the change of notation $q\to x$ and $p\to\lambda$).

\noindent $\bullet$ A trajectory $q(\cdot)$ for which there exist $p(\cdot),~u(\cdot)$, and $p^0$ such that
$(\qq(\cdot),p(\cdot),u(\cdot),p^0)$ satisfies all the conditions given by the PMP is called an {\em extremal trajectory} and the 4-uple $(\qq(\cdot),p(\cdot),u(\cdot),p^0)$ an {\em extremal} or, equivalently, an \emph{extremal lift of $\qq(\cdot)$}. Such an extremal is called {\em normal} if $p^0\neq0$ and {\em abnormal} if
$p^0=0$. It may happen that
 an extremal trajectory $\qq(\cdot)$ admits both a normal extremal lift
$(\qq(\cdot),p_1(\cdot),u(\cdot),p^0)$ and an abnormal
one $(\qq(\cdot),p_2(\cdot),u(\cdot),0)$.
In this case, we say that the extremal trajectory
$\qq(\cdot)$ is a {\em non-strict abnormal trajectory.} Note that (as in the finite-dimensional case) abnormal trajectories are  candidates for optimality regardless of the cost.
In the finite-dimensional case, they correspond to singularities of the constraint function, while here they correspond to singularities of the functional associating with a control $v(\cdot)$ the endpoint at time $T$ of the solution of $\dot q(t)=f(q(t),v(t))$, $q(0)=\qqin$.
It is worth noticing that abnormal extremals do not only appear in pathological cases, but they are often present in real-world applications, as for instance in the two examples presented in Sec.~\ref{secthree} and \ref{sectwo}.

\noindent $\bullet$ The PMP is only a necessary condition for optimality. It may very well happen that an extremal trajectory  is not optimal. The PMP can therefore provide several candidates for optimality, only some of which are optimal (or even none of them if the step of existence has not been verified, see Sec.~\ref{s-existence}).

\noindent $\bullet$ Since the equation for $p(\cdot)$ at point~{\bf ii)} of the PMP is linear, if  $(\qq(\cdot),p(\cdot),u(\cdot),p^0)$ is an extremal, then for every $\al>0$, $(\qq(\cdot),\al p(\cdot),u(\cdot),\al p^0)$ is an extremal  as well. As a consequence, some useful normalizations are possible. A typical normalization for normal extremals is to require  $p^0=-\frac12$ but other choices are also possible.

\noindent $\bullet$ When there is no final cost ($\phi=0$), the transversality condition $(\textbf{v})$ simplifies to:
\begin{equation}
\langle p(T),T_{\qq(T)}{\cal T}\rangle=0.
\label{tc}
\end{equation}
When the final point is fixed ($\target=\{\qqfin\}$), $T_{\qq(T)}\target$ is a zero-dimensional manifold and hence condition \eqref{tc} is empty.
When the final point is free ($\target=M$)
the transversality condition simplifies to $p(T)=p^0 d\phi(q(T))$. In local coordinates, we recover that $p(T)$ is proportional to the gradient of $\phi$ evaluated at the point $q(T)$.
Notice that, since $(p(T),p^0)\ne 0$, in this case one necessarily has $p^0\ne0$.

Table~\ref{tabext} gives a list of the possible extremal solutions of the PMP.

{
\begin{example}\label{exwavefunction}
As a general example in quantum control, we consider a dynamical system governed by Eq.~\eqref{eqwave} where the goal is to minimize at the fixed final time $T$
the cost $-|\langle \psi_{\mathrm{fi}}|\psi(T)\rangle |^2+\frac{1}{2}\int_0^T\sum_{j=1}^mu_j^2(t)dt$, where $\psi_{\mathrm{fi}}$ is a target state towards which we want to drive the system (up to a global phase)~\cite{gross}. A direct application of the PMP shows that the pre-Hamiltonian $\mathcal{H}$ can be expressed as:
$$
\mathcal{H}(\psi,\chi,u_j,p^0)=\Re(\langle\chi |\dot{\psi}\rangle)+\frac{p^0}{2}\sum_j u_j^2
$$
where the adjoint state, denoted here by $\chi$, is an abstract wave function which can be chosen so that it belongs to the unit sphere in $\mathbb{C}^N$, $\langle\chi |\chi\rangle =1$. Since $\psi$ and $\chi$ are complex-valued functions, the pre-Hamiltonian is defined through the real part of the scalar product between $\chi$ and $\dot{\psi}$. The standard definition used in Th.~\ref{t-PMP} can be found by introducing the real and imaginary parts of the wave functions. Using Eq.~\eqref{eqwave}, we deduce that
\begin{equation}\label{hamwave}
\mathcal{H}(\psi,\chi,u_j,p^0)=\Im(\langle\chi |H_0+\sum_ju_jH_j|\psi\rangle)+\frac{p^0}{2}\sum_j u_j^2
\end{equation}
which leads to
\begin{equation*}
i\dot\chi(t)=(H_0+\sum_{j=1}^m u_j(t)H_j)\chi(t),
\end{equation*}
i.e., $\chi$ also satisfies the Schr\"odinger equation. We stress that this condition is only true in the bilinear case. A specific equation has to be computed for other dynamics, such as, e.g., the Gross--Pitaevskii equation~\cite{bolzi2007}. The final condition $\chi(T)$ is given by the transversality condition {\bf v)} of the PMP:
\begin{equation}\label{finalwave}
\chi(T)=p^0\langle \psi_{\mathrm{fi}}|\psi(T)\rangle\psi_{\mathrm{fi}}.
\end{equation}
The maximization condition of the PMP leads to the  constraints
$$
\frac{\partial \mathcal{H}}{\partial u_j}(\psi,\chi,u_j,p^0)=0
$$
for $j=1,\dots,m$. A direct computation from Eq.~\eqref{hamwave} gives
$$
\Im(\langle\chi |H_j|\psi\rangle)+p^0 u_j=0.
$$
For the normal extremal with $p^0=-1$, we finally get
\begin{equation}\label{normalwaveu}
u_j=\Im(\langle\chi |H_j|\psi\rangle).
\end{equation}
\end{example}

\subsection{Use of the PMP}\label{secusePMP}

The application of the PMP is not so straightforward. Indeed, there are many conditions to satisfy and all of them are coupled. This section is aimed at describing how to use it in practice.

The following points should be followed first for normal extremals ($p^0<0$) and then for abnormal extremals ($p^0=0$).  In the first case, $p^0$ can be normalized to $-1/2$ since $p^0$ is defined up to a multiplicative positive factor. In the different steps, several difficulties (that are briefly mentioned) may arise. Most of them should be solved case by case, since they can be of different nature depending on the problem under study.
\bi
\iii[Step 1.] Use the maximization condition {\bf iii)} to express, when possible, the control as a function of the state and of the covector, i.e., $u=w(q,p)$.
Note that if we have $m$ controls (e.g.,  if $U$ is an open subset of $\R^m$) then the first-order maximality conditions
give $m$ equations for $m$ unknowns.
When
the maximization condition permits to express
 $u$ as a function of $q$ and $p$, we say that the control is {\em regular}, otherwise
the control is said to be {\em singular}. 
We may have regions where the control is regular and regions where it is singular. For singular controls, finer techniques have to be used to derive the expression of the control. These different points are discussed in the examples in Sec.~\ref{secthree} and \ref{sectwo}.

\iii[Step 2.] Insert the control found in the previous step into the Hamiltonian equations {\bf i)} and {\bf ii)}:
\begin{equation}
\begin{cases}
\dot \qq(t)=\frac{\partial \HHH
}{\partial p}(\qq(t),p(t),w(\qq(t),p(t)),p^0)\\
\dot p(t)=-\frac{\partial \HHH
}{\partial q}(\qq(t),p(t),w(\qq(t),p(t)),p^0).\label{cp}
\end{cases}
\end{equation}
In case the previous step provides a smooth $w(\cdot,\cdot)$, this is a well-defined set of $2n$ equations for $2n$ unknowns. Note, however, that the boundary conditions are given in a non-standard form since we know $\qq(0)$ but not $p(0)$. Instead of $p(0)$, we have a partial information on $\qq(T)$ and $p(T)$ depending on the dimension of $\target$ (see the next step to understand how these final conditions are shared between $\qq(T)$ and $p(T)$).
We then solve Eq.~\eqref{cp} for {\em fixed} $\qq(0)=\qqin$ and {\em any} $p(0)=\ppin \in T^\ast_{\qqin}M$.
Let us denote the solution as
\[
\qq(t;
\ppin,p^0),\qquad p(t;
 \ppin,p^0).
\]

We stress that when $w(\cdot,\cdot)$ is not regular enough,
solutions to the Cauchy problem \eqref{cp}
with $\qq(0)=\qqin$ and $p(0)=\ppin$ may fail to exist or to be unique.

\iii[Step 3.] Find $\ppin$ such that
\begin{equation}
\qq(T;
\ppin,p^0)\in\target.
\label{ttt}
\end{equation}
Note that if $\target$ is reduced to a point and $T$ is fixed, we get $n$ equations for $n$ unknown (the components of $\ppin$). If $T$ is free then an additional equation is needed. This condition is given by the relation {\bf iv)} in the PMP. If $\target$ is a $k$-dimensional submanifold of $M$ ($k\leq n$) then Eq.~\eqref{ttt} provides only $n-k$ equations and the remaining ones correspond to the transversality condition  {\bf v)} of the PMP.

\iii[Step 4.] If Eq.~\eqref{ttt} (together with the transversality condition and condition {\bf iv)} of the PMP if $T$ is free) has a unique solution $\ppin$  and if we have verified a priori the existence step, then the optimal control problem is solved. Unfortunately, in general there is no reason for Eq.~\eqref{ttt} to provide a unique solution. Indeed, the PMP is only a necessary condition for optimality. If several solutions are found, one should choose among them the best one by a direct comparison of the value of the cost. This is, in general, a non-trivial step, complicated by the difficulty of  solving explicitly Eq.~\eqref{ttt}. For this reason, several techniques have been developed to select the extremals. Among others, we mention the sufficient conditions for optimality given by Hamilton--Jacobi--Bellman theory and synthesis theory. We refer to \cite{enciclopedia} for a discussion.
In Example~1 (Section~\ref{secthree}) we are able to select the optimal solution without the use of sufficient conditions for optimality, while this is not the case in Example~2 (Section~\ref{sectwo}).
\ei

\begin{table}[tb]
\caption{List of the possible extremal solutions of the PMP.\label{tabext}}
\begin{tabular}{|c|c|}
\hline
{\bf Name}& Definition \\
\hline\hline
Extremal & 4-uple $(q(\cdot),p(\cdot),u(\cdot),p^0)$ solution of the PMP \\
\hline
Normal extremal & extremal with $p^0\neq 0$\\
\hline
Abnormal extremal & extremal with $p^0=0$ \\
\hline
Non-strict abnormal trajectory & Trajectory which admits both abnormal and normal lifts \\
\hline
Regular control & When the maximization condition of the PMP gives $u=\omega (q,p)$ \\
\hline
Singular control & When the control is not regular. \\
\hline
\hline
\end{tabular}
\end{table}

\begin{example}
We come back to the case of Example~\ref{exwavefunction}. We have shown with Eq.~\eqref{normalwaveu} that the maximization condition allows to express the $m$ controls $u_j$ as functions of $\psi$ and $\chi$ in the normal case. This situation therefore corresponds to Step 1 above where the control is regular. For abnormal extremals for which $p^0=0$, we get
\begin{equation}\label{eqabnormalwave}
\Im(\langle\chi |H_j|\psi\rangle)=0,
\end{equation}
and the control is singular because this relation does not give directly the expression of $u_j$.

Applying Step 2, we obtain in the regular situation the following coupled equations for $\psi$ and $\chi$:
\begin{equation}\label{eqdynoptwave}
\begin{cases}
i\dot\psi=[H_0+\sum_j\Im(\langle\chi |H_j|\psi\rangle)H_j]\psi \\
i\dot\chi=[H_0+\sum_j\Im(\langle\chi |H_j|\psi\rangle)H_j]\chi \\
\end{cases}
\end{equation}
with the boundary conditions $\psi(0)=\psi_{\textrm{in}}$ and Eq.~\eqref{finalwave}. In Step 3, we then solve Eq.~\eqref{eqdynoptwave} to find the initial condition $\chi(0)$ such that the final state $\chi(T)$ satisfies Eq.~\eqref{finalwave} at time $T$.

The numerical procedures used to select the initial condition $\chi(0)$, called \emph{shooting methods} in control theory,
are based on suitable adaptations of the Newton algorithm~\cite{bonnardbook2012,brysonbook}.
 Step 4 consists finally in comparing the cost of the different solutions found in Step 3.

In the abnormal case, we use the fact that Eq.~\eqref{eqabnormalwave} is satisfied in a non-zero time interval so the time derivatives of $\Im(\langle\chi |H_j|\psi\rangle)$ are also zero. The first time derivative leads to the 
$m$ relations
$$
\sum_{j=1}^mu_j\Re(\langle\chi |[H_k,H_j]|\psi)=\Re(\langle\chi |[H_0,H_k]|\psi\rangle),
$$
with $k=1,\dots,m$. This linear system can be expressed in a more compact form as
$$
Ru=s,
$$
where $R$ is a $m\times m$ matrix with  elements $R_{kj}=\Re(\langle\chi |[H_k,H_j]|\psi)$ and $s$ a vector of coordinates $s_k=\Re(\langle\chi |[H_0,H_k]|\psi\rangle)$. We deduce that the control $u$ is given as a function of $\psi$ and $\chi$ as $u=R^{-1}s$. If this system is singular then the second time derivative has to be used.
This  is the case, e.g., for $m=1$, when a 
further constraint has to be fulfilled, namely, $\Re(\langle\chi |[H_0,H_{1}]|\psi\rangle)=0$. From the derivation of $u$, we then apply Steps 2 and 3 to the abnormal extremals.

\end{example}

\section{Gradient-based optimization algorithm}\label{sectiongradient}
The aim of this section is to introduce a first order gradient-based optimization algorithm 
based on the PMP. We first derive the necessary conditions of the PMP in the case of a fixed control time without any constraint on the final state and on the control. This construction is known in control theory as the weak PMP. Note that we consider only the case of regular control. Iterative algorithms can be deduced from these conditions. In a second step, we apply this idea to quantum control and we show how a gradient-based optimization algorithm, GRAPE~\cite{KHANEJA2005}, can be designed from this approach.
\subsection{The Weak Pontryagin Maximum Principle}\label{WPMP}
We consider a control system whose dynamics are governed by Eq.~\eqref{control}, when the final state is free and the control is unconstrained. The objective is to solve a control problem in the Approach B, as defined in Sec.~\ref{secformu} with a fixed control time $T$. We recall that the cost functional to minimize 
can be expressed as
$$
J(u(\cdot))=\int_0^T f^0(\qq(t),u(t))\,dt+d(\target,\qq(T)).
$$

Considering the evolution equation~\eqref{control} as a dynamical constraint (in infinite dimension), in order to apply (formally)
the Lagrange multiplier rule for normal extremals, we introduce the functional
\begin{align}\label{jbar}
\Lambda(&p(\cdot),u(\cdot))= d(\target,\qq(T))+\int_0^T f^0(\qq(t),u(t))\,dt \nonumber \\
&  +\int_0^T\langle p(t), \dot{q}(t)-f(q(t),u(t))\rangle\,dt.
\end{align}
We stress that the Lagrange multiplier $p(\cdot)$ is here a function on $[0,T]$.
Integrating by parts Eq.~\eqref{jbar}, we obtain
\begin{align*}
\Lambda(&p(\cdot),u(\cdot))= d(\target,\qq(T))+\langle p(T),q(T)\rangle\\
 & -\langle p(0),q(0)\rangle \\& -\int_0^T (H(q(t),p(t),u(t))+\langle\dot{p}(t),q(t)\rangle )\,dt,
\end{align*}
with
\begin{equation}\label{pmpweak}
H(q,p,u)=\langle p,f(q,u)\rangle -f^0(\qq,u).
\end{equation}
Note that the scalar function $H$ has the same expression (with $p^0=-1$) as the pre-Hamiltonian $\mathcal{H}$ introduced in Sec.~\ref{ss-PMP} for the statement of the PMP. Since there is no constraint on the control law, i.e., $u(t)\in\mathbb{R}^m$ for any time $t$, we consider the variation $\delta\Lambda$ in $\Lambda$ at first order due to the variation $\delta u$ of $u$. This change of control induces a variation of the trajectory $\delta q(t)$ with $\delta q(0)=0$, the initial point being fixed. Note that the adjoint state $p$ is not modified. We arrive at:
\begin{eqnarray}\label{jbar2}
& & \delta\Lambda= \langle\left.\frac{\partial d({\cal T},q)}{\partial q}\right|_{q=q(T)}+p(T),\delta q(T)\rangle
 \nonumber \\
& & -\int_0^T [\langle \frac{\partial H}{\partial q}+\dot{p},\delta q\rangle +\frac{\partial H}{\partial u}\delta u]\,dt.
\end{eqnarray}
A necessary condition for $\Lambda$ to be an extremum is $\delta\Lambda=0$ for any  variation $\delta u$. A solution is given by taking an  adjoint state $p$ satisfying
\begin{equation}\label{ff1}
\dot{p}=-\frac{\partial H}{\partial q},
\end{equation}
the final boundary condition
\begin{equation}\label{ff2}
p(T)=-\left.\frac{\partial d({\cal T},q)}{\partial q}\right|_{q=q(T)},
\end{equation}
and requiring
\[\frac{\partial H}{\partial u}=0\qquad\mbox{on } [0,T].\]

As it could be expected, we find here a weak version of the equations of the PMP introduced in Sec.~\ref{ss-PMP} where the maximization of the pre-Hamiltonian is replaced by an extremum condition given by the partial derivative with respect to $u$. We point out that this approach works if the set $U$ is open.
\subsection{Gradient-based optimization algorithm}\label{secgrapeclass}
The set of nonlinear coupled differential equations can be solved numerically by an iterative algorithm. The basic idea used in such algorithms can be formulated as follows. Assume that a control $u(\cdot)$ sufficiently close 
to the optimal solution is known.
If $p(\cdot)$ satisfies Eq.~\eqref{ff1} and \eqref{ff2} then we deduce from Eq.~\eqref{jbar2}
that
\[\delta\Lambda=-\int_0^T \frac{\partial H}{\partial u}\delta u\,dt.\]
This suggests that a better control can be achieved with the choice $\delta u=\epsilon \frac{\partial H}{\partial u}$ where $\epsilon$ is a small positive
parameter.
The iterative algorithm is then described by the following steps.
\begin{enumerate}
\item Choose a guess control $u(\cdot)$.
\item Propagate forward the state of the system $q$ from $\dot{q}=f(q,u)$, with the initial condition $q(0)=q_0$.
\item Propagate backward the adjoint state of the system from $\dot{p}=-\frac{\partial H}{\partial q}$, with the final condition $p(T)=-\left.\frac{\partial d({\cal T},q)}{\partial q}\right|_{q=q(T)}
$.
\item Compute the correction $\delta u$ to the control law, $\delta u(t)=\epsilon \frac{\partial H}{\partial u}$  where $\epsilon>0$ is a small parameter.
\item Define the new control $u\mapsto u+\delta u $.
\item Go to step 2 and repeat until a given accuracy is reached.
\end{enumerate}
This algorithm is an example of first-order gradient-based optimization algorithm. By construction, it converges towards an extremal control of $J$ which is not, in general, a global 
minimum solution of the control problem, but only a local one.
Note that several numerical details are hidden in the description of this method. Among others, we mention the choice of the guess control which allows to reach a good local solution and the determination of the parameter $\epsilon$. This latter must be sufficiently small to remain in the first order approximation, but large enough to reduce the number of iterations and therefore the computational time. We refer the reader to standard numerical optimization textbooks to address these issues~\cite{brysonbook}.

This approach can be directly applied to quantum systems. The bilinearity of quantum dynamics allows to simplify the different terms used in the algorithm. We consider a quantum system whose dynamics are governed by the Schr\"odinger equation
$$
i\dot{\psi}(t)=(H_0+u(t)H_1)\psi(t).
$$
The goal of the control process is to bring the system from $\psi_{\textrm{in}}$ towards  $\psi_{\textrm{fi}}$ in a fixed time $T$. The control problem aims at minimizing 
the
cost
$$
J=\frac{1}{2}\int_0^Tu(t)^2dt-|\langle \psi_{\rm fi}|\psi(T)\rangle |^2.
$$
In the normal case, the pre-Hamiltonian $H$ can be expressed as:
$$
H=\Re[\langle \chi|\dot{\psi}\rangle ]-\frac{u^2}{2}=\Im[\langle \chi|H_0+uH_1|\psi\rangle]-\frac{u^2}{2},
$$
where the wave function $\chi$ is the adjoint state of the system. We thus deduce that the gradient on which the iterative algorithm is based is given by
\begin{equation}\label{eqgrad}
\frac{\partial H}{\partial u}=\Im[\langle \chi|H_1|\psi\rangle]-u(t).
\end{equation}
We find with Eq.~\eqref{eqgrad} the standard control correction used in the GRAPE algorithm in quantum control~\cite{KHANEJA2005}.

\section{Example 1: A three-level quantum system with complex controls}\label{secthree}
In this section, we mainly use the results of \cite{boscain-q1}, see also \cite{boscainres,boscain-nonisotropic,sugny08}. Note, however, that original results concerning the selection of the best extremal among all the possible solutions are presented.
\subsection{Formulation of the quantum control problem}\label{s:excexc1}
We consider a three-level quantum system whose dynamics are governed by the Schr\"odinger equation. The system is described by a pure state $\psi(t)$ belonging
to a three-dimensional complex Hilbert space.
The system is
characterized, in the absence of external fields,
 by three energy levels $E_1$, $E_2$, and $E_3$ and is
controlled by the Pump and the Stokes pulses which couple, respectively, states one and two and states two and three~\cite{vitanov2001}. Note that there is no direct coupling between levels one and three. The time evolution of $\psi(t)$ is given by
$$
i\dot\psi(t)=H(t
)\psi(t),
$$
where
$$
H(t
)=\begin{pmatrix}
E_1 & \Omega_1(t) & 0 \\
\Omega_1^*(t) & E_2 & \Omega_2(t) \\
0 & \Omega_2^*(t) & E_3 \\
\end{pmatrix}.
$$
Here $\Omega_1(t),\Omega_2(t)\in \mathbb{C}$ are the two time-dependent complex control parameters. We denote by $\psi_1(t)$, $\psi_2(t)$, and $\psi_3(t)$ the coordinates of $\psi(t)$, that is, $\psi(t)=(\psi_1(t),\psi_2(t),\psi_3(t))$.
They
satisfy
$$
|\psi_1(t)|^2+|\psi_2(t)|^2+|\psi_3(t)|^2=1,
$$
leading to $M=S^5$, a manifold of real dimension 5. The goal of the control process is to transfer population from the first eigenstate to the third one in a fixed time $T$. In other words, the aim is to find a trajectory in $M$ going from the {submanifold $|\psi_1|^2=1$ to the one with $|\psi_3|^2=1$.} The system is completely controllable thanks to Point~2 in Proposition~\ref{p-controllability}
and the approach A can be chosen. The optimal control problem is defined through the cost functional
$$
C={ \int_0^T(|\Omega_1(t)|^2+|\Omega_2(t)|^2)dt,}
$$
to be minimized. The cost $C$ can be interpreted as the energy of the control laws  used in the control process. We consider the specific case in which the control parameters are in resonance with the energy transition.
More precisely, we assume that the pulses $\Omega_1(t)$ and $\Omega_2(t)$ can be expressed as
$$
\begin{cases}
\Omega_1(t)=u_1(t)e^{i(E_2-E_1)t}\\
\Omega_2(t)=u_2(t)e^{i(E_3-E_2)t}
\end{cases}
$$
with $u_1(t),u_2(t)\in \R$ to be optimized. Note that this assumption is not restrictive since it can be shown that the resonant case corresponds to the optimal solution~\cite{boscain-q1,boscainres}. The uncontrolled part, called the drift, together with the imaginary unit in the Schr\"odinger equation, can be eliminated through a unitary transformation $Y(t)$ given by $$Y(t)=\textrm{diag}(e^{-iE_1t},e^{-i(E_2t+\pi/2)},e^{-i(E_3t+\pi)}).$$

Defining a new wave function $x$ such that $\psi(t)=Y(t) x(t)$, we obtain that $x(t)$ solves the
Schr\"odinger equation
$$
i\dot x(t)=H'(t)x (t),
$$
where $H'=Y^{-1}HY-iY^{-1}\dot{Y}$. Since $Y$ only modifies the phases of the coordinates, $\psi(t)$ and $x(t)$ correspond to the same population distribution, i.e., setting $x=(x_1,x_2,x_3)$ we have $|x_j(t)|^2=|\psi_j(t)|^2$, $j=1,2,3$.

Computing explicitly $H'$  we arrive at
\begin{equation}
\begin{pmatrix}
\dot{x}_1 \\
\dot{x}_2 \\
\dot{x}_3
\end{pmatrix} =\begin{pmatrix}
0 & -u_1(t) & 0\\
u_1(t) & 0 & -u_2(t) \\
0 & u_2(t) & 0
\end{pmatrix}\begin{pmatrix}
{x}_1 \\
{x}_2 \\
{x}_3
\end{pmatrix}.
\label{uuu}
\end{equation}
Without loss of generality, the optimal control problem can be restricted to the submanifold $S^2\subset M$ defined by $\{(x_1,x_2,x_3)\in \R^3\mid x_1^2+x_2^2+x_3^2=1\}$. This statement is trivial if the initial condition belongs to $S^2$ (i.e., if $x(0)$ is real).
Otherwise,  a straightforward change of phase of the coordinates $\psi_k$ allows to come back to this condition. Equation~\eqref{uuu} can be expressed in a more compact form as
\begin{equation}
\dot x=u_1(t) F_1(x)+u_2(t) F_2(x),\label{vvv1}
\end{equation}
where
$$
F_1(x)=\left(\ba{c}-x_2\\x_1\\0\ea\right),~
F_2(x)=\left(\ba{c}0\\-x_3\\x_2\ea\right).
$$
Notice that $F_1$ and $F_2$ are two vector fields defined on the sphere $S^2$ representing, respectively, a rotation along the $x_3$ axis and along the $x_1$ axis.
 Since $\Omega_1(t)$ and $\Omega_2(t)$ differ from
$u_1(t)$ and $u_2(t)$  only for phase factors, the minimization problem becomes
\begin{equation}
C={ \int_0^T(u_1(t)^2+u_2(t)^2)dt}\to\min,
\label{vvv2}
\end{equation}
with $T$ fixed. Concerning initial and final conditions, since the goal is to go from the submanifold $|\psi_1|^2=1$ to the submanifold
$|\psi_3|^2=1$, we can assume without loss of generality that $x(0)=(1,0,0)$ (again, a straightforward change of coordinates allows to come back to this condition
otherwise).
Since we are now restricted to real variables, the target is  $\target=\{(0,0,+1),(0,0,-1)\}$.
Now, being the target made of two points only, one should compute separately the optimal trajectories going from $(1,0,0)$
to $(0,0,1)$ and those going from $(1,0,0)$
to $(0,0,-1)$. Finally between all these trajectories, one should take the ones having the smallest cost.
Because of the symmetries of the system, the two families of optimal trajectories have precisely the same cost (this will be clear in the explicit computations later on). As a consequence, without loss of generality, we can  fix the final condition as $x(T)=(0,0,1)$.

The problem \eqref{vvv1}--\eqref{vvv2} with fixed initial and final conditions is actually a celebrated problem in OCT called  the Grushin model on the sphere~\cite{grushin-sphere,boscain-nonisotropic,grusin}.

\subsection{Existence}
For convenience, let us re-write the optimal control problem as follows:\\

\noindent{\bf Problem} P$^{{\rm Grushin}}(T)$
\begin{align}
&\dot x=u_1(t) F_1(x)+u_2(t) F_2(x),\nn\\
&{ \int_0^T(u_1(t)^2+u_2(t)^2)dt}\to\min,\quad T\mbox{ fixed},\nn\\
&x(0)=(1,0,0),\quad x(T)=(0,0,1),\nn\\
&u_1,u_2\in \mathcal{U},~U=\mathbb{R}.
\nn
\end{align}
To prove the existence of solutions to P$^{{\rm Grushin}}(T)$, one could be tempted to use Propositions~\ref{p-teschio-quantico} and~\ref{p-teschio}. However
$u_1$ and $u_2$ take values in $\R$ and, hence, the second hypothesis of the proposition is not verified.

Instead, we are going to use the following fact.
\bp
\label{mufcona} If $u_1(t),u_2(t)$ are optimal controls for P$^{{\rm Grushin}}(T)$, then $u_1(t)^2+u_2(t)^2$  is almost every\-where constant and positive on $[0,T]$.
Moreover, for every $\al>0$, we have that $\al u_1(\al t),$ $\al u_2(\al t)$ are optimal controls for {\rm P}$^{{\rm Grushin}}(T/\alpha)$.
\ep
To prove this proposition we first state a general lemma for driftless systems.
\begin{lemma}
\label{melenzane}
Consider a  control system of the form $\dot x=\sum_{j=1}^m u_j(t)F_j(x)$, where $x\in M$ and $u_j\in \mathcal{U}$, $j=1,\ldots,m$. Then
any admissible trajectory $x(\cdot)$ defined on $[0,T]$ and corresponding to controls $u_j(\cdot)$, $j=1,\ldots,m$, is a reparameterization  of an admissible trajectory $\bar x(\cdot)$ defined on the same time interval  whose controls $\bar u_j(\cdot)$ satisfy a.e. $\sqrt{\sum_{j=1}^m \bar u_j(t)^2}=L/T$, where $L=\int_0^T\sqrt{\sum_{j=1}^m  u_j(t)^2}dt$.  Alternatively, $x(\cdot)$ is a reparameterization  of a trajectory $\bar x(\cdot)$ defined on $[0,L]$  whose controls satisfy a.e. $\sqrt{\sum_{j=1}^m \bar u_j(t)^2}=1$.
\end{lemma}

We recall that, given an admissible trajectory $x:[0,T]\to M$, $T>0$, corresponding to controls $u_j(\cdot)$, $j=1,\ldots,m$,
a {\em reparameterization} of $x(\cdot)$ is a trajectory $\bar x(\cdot)=x(\tau(\cdot))$
with $\tau:[0,{\bar T}]\to[0,T]$ a function such that $\frac{d}{dt}\tau(t)>0$ a.e. on $[0,\bar T]$.
Such trajectory is  defined on $[0,{\bar T}]$.
Lemma~\ref{melenzane}
is
a consequence of the fact that
for a.e. $t\in[0,\bar T]$ we have
\begin{align}
\dot{\bar{x}}(t)&=\frac{d}{dt}x(\tau(t))=\dot x(\tau(t)) \dot\tau(t)\nonumber\\
&=\Big(\sum_{j=1}^m u_j(\tau(t))F_j(x(\tau(t))\Big) \dot\tau(t) \nonumber\\
&=\sum_{j=1}^m \Big(u_j(\tau(t)) \dot\tau(t) \Big)F_j(\bar x(t)),\nonumber
\end{align}
from which it follows that $\bar x(\cdot)$ is admissible and corresponds to  controls $u_j(\tau(\cdot))\dot \tau(\cdot)$,  $j=1,\ldots,m$.

In the following, for the optimal control problem under study, it is  convenient to normalize $T$ in such a way that $u_1(t)^2+u_2(t)^2=1$. Usually, when one makes this choice
the trajectories are said to be {\em parameterized by arc length}. If the objective is to reach the target at time $T'$, it is sufficient to use the controls $(\al\, u_1(\al \, t),\al\, u_2(\al \,t))$, where $\al=\frac{T}{T'}$.

When $T$ is fixed in such a way that $u_1(t)^2+u_2(t)^2=1$, we call problem P$^{{\rm Grushin}}(T)$   simply P$^{{\rm Grushin}}$.

Notice that, as a consequence of Lemma~\ref{melenzane}, if $T$ is not fixed then
P$^{{\rm Grushin}}(T)$ has no solution.
Indeed,  assume by contradiction that $x(\cdot)$, defined on $[0,T_0]$ and corresponding to controls $u_j(\cdot)$, $j=1,2$,  is a minimizer  for $T$ free.
Let $c\in(0,1)$. The trajectory corresponding to controls $\bar u_j(\cdot)=u_j( c\, \cdot)c$, $j=1,2$, is a reparameterization of $x(\cdot)$ reaching the final point at time $\bar T=T_0/c$ with a cost
\begin{align}
\int_0^{\frac{T_0}{c}} \sum_{j=1}^2 u_j(ct)^2c^2 dt&=c\int_0^{T_0} \sum_{j=1}^2 u_j(t)^2 dt,\nonumber\\
&<\int_0^{T_0} \sum_{j=1}^2 u_j(t)^2 dt,\nonumber
\end{align}
which leads to a contradiction.

\begin{proof}[Proof of Proposition~\ref{mufcona}]
Let us define
\begin{align}
L(u(\cdot))&=\int_0^T\sqrt{u_1(t)^2+u_2(t)^2}dt,\nonumber\\
E(u(\cdot))&=\int_0^T(u_1(t)^2+u_2(t)^2)dt.\nonumber
\end{align}
We are going to use the Cauchy--Schwarz inequality
$\langle f,g\rangle ^2\leq\|f\|^2\|g\|^2$ (with equality holding if and only if $f$ and  $g$ are proportional), which holds in any Hilbert space.
This inequality simply tells that the scalar product of two vectors is less than or equal to the product of the norms of the two vectors (with equality holding iff the two vectors are collinear).
In particular, this can be used in the space $L^2([0,T],\mathbb{R})$ of measurable functions $f:[0,T]\to \R$ with $\int_0^Tf(t)^2\,dt<\infty$.  Namely
\begin{align}
\Big(\int_0^T f(t)g(t)\,dt\Big)^2&\leq\int_0^T f(t)^2\,dt \int_0^T g(t)^2\,dt  \nonumber
\end{align}
(with equality holding iff $f \propto g$, a.e.).
 Now, let  $f(t)=\sqrt{u_1(t)^2+u_2(t)^2}$ and $g(t)=1$ for $t\in[0,T]$. Notice that $f,g\in \mathcal{U}\subset L^2([0,T],\mathbb{R})$. We have
 \begin{align}
L(u(\cdot))^2&\leq E(u(\cdot))T \nonumber
\end{align}
 (with equality holding iff $u_1(t)^2+ u_2(t)^2=\mathrm{const}$ a.e.).
  Now, let $u(\cdot)$ be a minimizer of $E$ defined on $[0,T]$. Assume by contradiction that $u_1(t)^2+u_2(t)^2$ is not a.e. constant. Then $L(u(\cdot))^2< E(u(\cdot))T$. Let $\bar x(\cdot)$ be an admissible trajectory defined on $[0,T]$ corresponding to controls satisfying $\sqrt{\bar u_1(t)^2+\bar u_2(t)^2}=L(u(\cdot))/T$ a.e., of which $x(\cdot)$ is a reparameterization as in Lemma \ref{melenzane}.
One immediately checks that $L(u(\cdot))=L (\bar u(\cdot))$.
 For this trajectory we have
$$
E(\bar u(\cdot))T =L(\bar u(\cdot))^2=L(u(\cdot))^2<E(u(\cdot))T,
 $$
contradicting the fact that $u(\cdot)$ is a minimizer of $E$.
\end{proof}

Now we have the following.
\bp
\label{fak-ka}
The problem {\rm P}$^{{\rm Grushin}}$   is equivalent to the problem of minimizing time with  the constraint on the controls $u_1(t)^2+u_2(t)^2\le 1$.
\ep
\begin{proof}
Since for  P$^{{\rm Grushin}}$ we have
$u_1(t)^2+u_2(t)^2=1$
a.e., then
$$
\int_0^T
(u_1(t)^2+u_2(t)^2)dt=T.
$$
Hence P$^{{\rm Grushin}}$ is equivalent to the problem of minimizing $T$ with the constraint on the controls
$u_1(t)^2+u_2(t)^2=1$.  To conclude the proof, let us show that a trajectory corresponding to controls for which the condition
\begin{equation}
u_1(t)^2+u_2(t)^2=1\quad \mathrm{a.e.} \label{tabaccone}
\end{equation}
is not satisfied cannot be optimal for the time-optimal control problem mentioned in the statement.
Actually if \eqref{tabaccone}  is not satisfied then  $L=\int_0^T\sqrt{u_1(t)^2+u_2(t)^2}<T$ and  an arc length  reparameterization of the trajectory reaches the target  in time exactly $L$.
 \end{proof}

Let us now go back to the problem of existence of optimal trajectories for {\rm P}$^{{\rm Grushin}}$.
Thanks to Proposition \ref{fak-ka}, {\rm P}$^{{\rm Grushin}}$ can be equivalently recast as a time-optimal control problem
with controls in the convex and compact set $U=\{(u_1,u_2)\in\R^2\mid  u_1^2+u_2^2\leq1\}$. We can then apply Propositions~\ref{p-teschio-quantico} and~\ref{p-tempomin} and
deduce the existence of an optimal trajectory for
P$^{{\rm Grushin}}$ (as well as P$^{{\rm Grushin}}(T)$ for every $T>0$).

\subsection{Application of the PMP}

}
Before applying the PMP, it is convenient to reformulate the problem in spherical coordinates. Indeed, one can prove the following statement.\\[2mm]
{\bf Claim.} Consider an optimal control problem as in the statement of the PMP (Theorem~\ref{t-PMP}). If all admissible trajectories starting from $\qqin$ are contained in a submanifold of $M$ of dimension strictly smaller than $n$, then each
 admissible
trajectory
has  an abnormal extremal lift.
\\[2mm]
 This property can be qualitatively justified as follows.
We have already mentioned that
 abnormal trajectories
 correspond to singularities of the functional associating with a control law
 the endpoint of the  corresponding controlled trajectory.
 If all admissible trajectories starting from $\qqin$ are contained in a proper submanifold of $M$, then the endpoint functional is everywhere singular, meaning that
 each admissible trajectory
 is abnormal.

As a consequence,  since in our case all trajectories are contained in the sphere $S^2$, if we apply the PMP in $\R^3$ all optimal trajectories admit an abnormal extremal lift. This creates additional difficulties that can be avoided working directly on $S^2$ in spherical coordinates.

Let us introduce the coordinates $(\theta,\varphi)$ as displayed in Fig.~\ref{figgrunew} such that:
$$
x_1=\sin\theta\cos\varphi, \quad
x_2 =\cos\theta,\quad
x_3 =\sin\theta\sin\varphi.
$$
\begin{center}
\begin{figure}
\begin{center}
\includegraphics[scale=0.8]{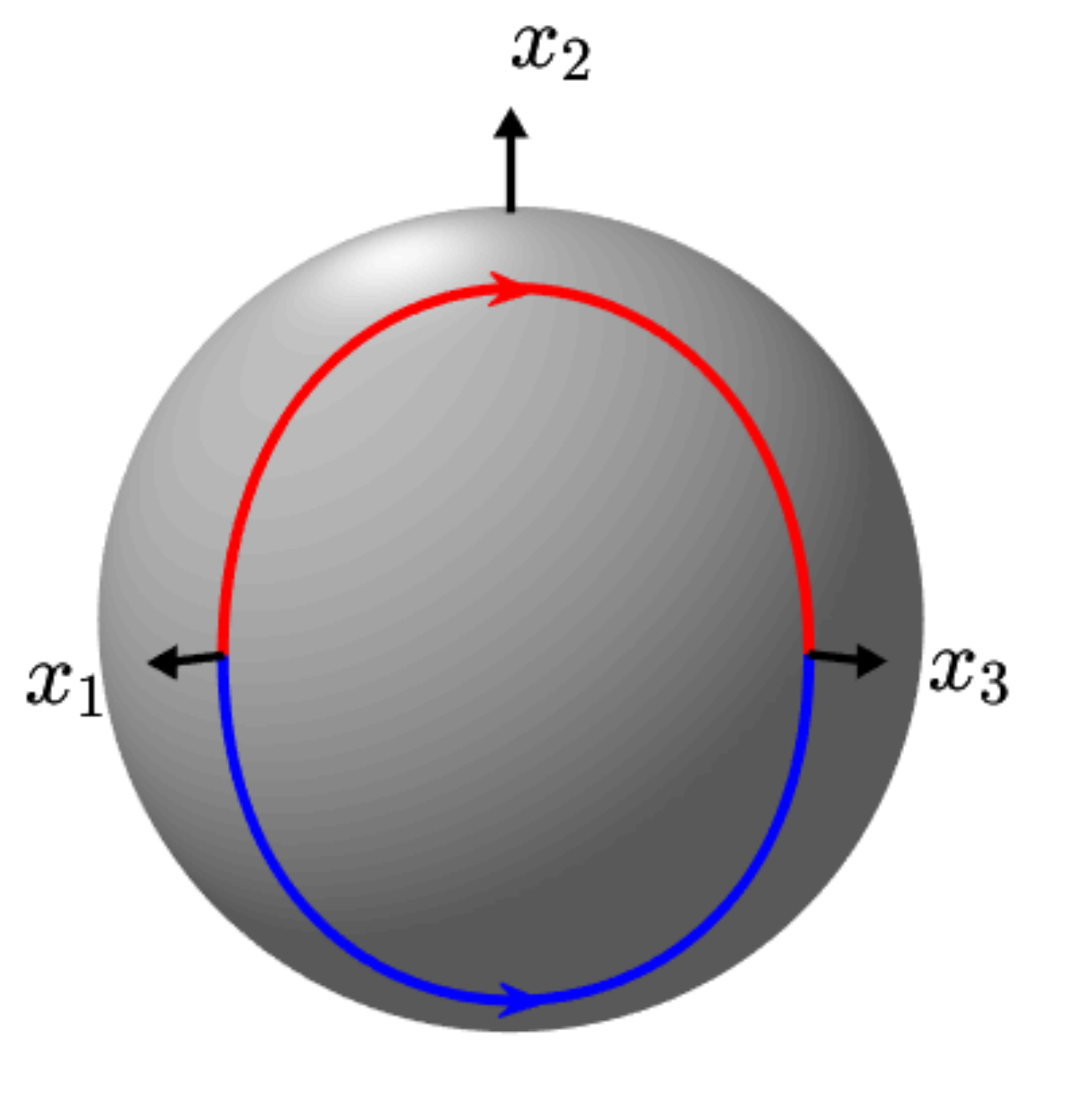}
\end{center}
\caption{(Color online) Picture of the sphere with the spherical coordinates $\theta$ and $\varphi$. \label{figgrunew}}
\end{figure}
\end{center}
In these coordinates, the starting point $x(0)=(1,0,0)$ and the final point $x(T)=(0,0,1)$ become $(\theta,\varphi)(0)=(\pi/2,0)$, $(\theta,\varphi)(T)=(\pi/2,\pi/2)$. Notice that these coordinates are singular for $\theta=0$ and $\theta=\pi$ but such a singularity does not create any problem, as can be checked by using a second system of coordinates around the singularity. The control system takes then the form
\begin{equation}\label{eqthetaphi}
\begin{cases}
\dot{\theta}=-u_1(t)\cos\varphi+u_2(t)\sin\varphi \\
\dot{\varphi}=\cot(\theta)(u_1(t)\sin\varphi+u_2(t)\cos\varphi).
\end{cases}
\end{equation}
It can be simplified by using the controls
$v_1$ and $v_2$ defined by
\begin{equation}
\begin{cases}
v_1=-u_1\cos\varphi+u_2\sin\varphi\\
v_2=u_1\sin\varphi+u_2\cos\varphi,
\end{cases}
\label{etmoide}
\end{equation}
which do not modify the expression of the cost $C$ since $u_1^2+u_2^2=v_1^2+v_2^2$.
The control system becomes now
$$
\left(\ba{c} \dot \theta\\ \dot\varphi\ea\right)=v_1(t)X_1(\theta,\varphi)+v_2(t)X_2(\theta,\varphi),$$
where
\[X_1(\theta,\varphi)=\left(\ba{c} 1
\\ 0\ea\right),~ X_2(\theta,\varphi)=\left(\ba{c} 0\\ \cot(\theta)\ea\right).
\]
Let us now apply the PMP.
Set $q=(\theta,\varphi)$ and let $p=(p_\theta,p_\varphi)$. The pre-Hamiltonian~\eqref{HAMHAMHAM} has the form
\begin{align*}
\HHH(&q,p,v,p^0)\\
&=v_1\, \langle p, X_1(q)\rangle+v_2\, \langle p, X_2(q) \rangle  +p^0(v_1^2+v_2^2)\\
&=v_1p_\theta+v_2 p_\varphi \cot(\theta)+p^0(v_1^2+v_2^2).
\end{align*}
We consider the steps of Section~\ref{secusePMP} first for abnormal ($p^0=0$) and then for normal ($p^0=-\frac12$) extremals.\\

\noindent\textbf{Step 1.} In this step, we have to apply the maximization condition to find the control as a  function of $q$ and $p$.
Since  the controls are unbounded and the Hamiltonian is
concave,
the maximization condition is equivalent to:
\begin{equation}
\begin{cases}
\frac{\partial \HHH}{\partial v_1}(q(t),p(t),v(t),p^0)\equiv0\\
\frac{\partial \HHH}{\partial v_2}(q(t),p(t),v(t),p^0)\equiv0.
\label{femore}
\end{cases}
\end{equation}

For abnormal extremals, we obtain
$$
\begin{cases}
\langle p(t), X_1(q(t))\rangle=p_\theta(t)\equiv0\\
\langle p(t), X_2(q(t))\rangle=p_\varphi(t)\cot(\theta(t))\equiv0.
\end{cases}
$$
These conditions do not permit to obtain the control as a function of $q$ and $p$. Hence, for this problem, abnormal extremals correspond to singular controls. Since $p$ and $p^0$ cannot be simultaneously  zero, the only possibility to have an abnormal extremal is that $\theta(t)\equiv\pi/2$ on $[0,T]$.  In this case, from Eq.~\eqref{eqthetaphi}, we deduce that $\varphi(t)$ should be constant. As a consequence, an abnormal extremal trajectory starting from the initial condition $(\theta,\varphi)(0)=(\pi/2,0)$ will never reach the final condition $(\theta,\varphi)(T)=(\pi/2,\pi/2)$ and we can disregard these trajectories.

For normal extremals, condition \eqref{femore} gives
\begin{equation}
\begin{cases}
v_1(t)=\langle p(t), X_1(q(t))\rangle=p_\theta(t)\\
v_2(t)=\langle p(t), X_2(q(t))\rangle=p_\varphi(t)\cot(\theta(t)).
\label{mandibola}
\end{cases}
\end{equation}
Hence, we obtain the controls as a function of $q$ and $p$ and we can conclude that normal extremals correspond to regular controls.

\noindent\textbf{Step 2.} Let us insert \eqref{mandibola} into the Hamiltonian equations {\bf i)} and {\bf ii)} of Theorem~\ref{t-PMP}. We have  to consider  the case $p^0=-1/2$ only. We obtain:
\begin{align}
\dot \theta(t)&=\frac{\partial\HHH}{\partial p_\theta}(q(t),p(t),v(t),-1/2)\nn\\
&=v_1(t)=p_\theta(t),\label{sr-1}\\
\dot p_\theta(t)&=-\frac{\partial\HHH}{\partial \theta}(q(t),p(t),v(t),-1/2)\nn\\
  &=v_2(t) p_\varphi(t)(1+\cot(\theta(t))^2)\nn\\
  &= p_\varphi(t)^2 \cot(\theta(t))(1+\cot^2(\theta(t))),\label{sr-2}\\
\dot \varphi(t)&=\frac{\partial\HHH}{\partial p_\varphi}(q(t),p(t),v(t),-1/2)\nn \\
&=v_2(t) \cot(\theta(t))= p_\varphi(t) \cot^2(\theta(t)),\label{sr-3}\\
\dot p_\varphi(t)&=-\frac{\partial\HHH}{\partial \varphi}(q(t),p(t),v(t),-1/2)=0.\label{sr-4}
\end{align}
Equation \eqref{sr-4} tells us that  $p_\varphi$ is a constant of the motion, denoted by $a$ in the following. We are then left with the differential equations
\begin{align}
\dot \theta= p_\theta,\quad
\dot p_\theta=a^2 \cot(\theta)(1+\cot^2(\theta)).\label{omero}
\end{align}
Once these are solved, $\varphi$ is obtained by integrating in time equation~\eqref{sr-3} which now has the form
\begin{align}
\dot \varphi= a  \cot^2(\theta).
\label{tibia}
\end{align}
Equations \eqref{omero} and \eqref{tibia} should be solved for every value of $a\in\R$ with the initial conditions
$$
\theta(0)=\pi/2,\quad \varphi(0)=0,\quad p_\theta(0)=\pm1.
$$
The last condition comes from the property that the maximized Hamiltonian is now fixed to $\frac12$ (corresponding to the choice of taking $T$ in such a way that optimal trajectories are parameterized  by arc length). More precisely:
\begin{align}
\frac12={}&\HHH(q(t),p(t),v(t),-1/2) \nn\\
={}&v_1(t) p_\theta(t)+v_2(t) a \cot(\theta(t))\nn\\
&-\frac12(v_1(t)^2+v_2(t)^2)\nn\\
={}&\frac12(p_\theta(t)^2+a^2\cot^2(\theta(t))).\nn
\end{align}
Requiring this condition at $t=0$, one gets  $p_\theta(0)=\pm1$.

The system of equations~\eqref{omero} can be solved using again that the maximized Hamiltonian is equal to $\frac12$ which implies
$$
\dot\theta(t)^2=1-a^2\cot^2(\theta(t)).
$$
Using a separation of variables, we arrive (with the initial condition $\theta(0)=\pi/2$) at
\begin{equation*}
\theta(t)=\begin{cases}
 \arccos\left( \frac{\sin(\sqrt{1+a^2} \,t)}{\sqrt{1+a^2}} \right)&\hspace{-.3cm}\mbox{ if } p_\theta(0)=-1,\\
\pi-\arccos\left( \frac{\sin(\sqrt{1+a^2}\, t)}{\sqrt{1+a^2}} \right)&\hspace{-.3cm}\mbox{ if } p_\theta(0)=1.
\end{cases}
\label{eq-thetafinal}
\end{equation*}
These expressions lead to a simple formula for $x_2(t)$, namely,
\begin{align}
x_2(t)=\pm\left( \frac{\sin(\sqrt{1+a^2} \,t}{\sqrt{1+a^2}} \right).
\label{malleolo-2}
\end{align}
As already explained, the expression for $\varphi$ can be obtained by integrating in time equation \eqref{tibia} using the expression \eqref{eq-thetafinal} with the initial condition $\varphi(0)=0$. The final result is easier expressed in Cartesian coordinates:
\begin{align}
x_1(t)={}&\frac{a \sin (a \,t) \sin \left(\sqrt{1+a^2}
   \,t\right)}{\sqrt{1+a^2}}\nn \\
   &+\cos (a \,t) \cos
   \left(\sqrt{1+a^2} \,t\right)\label{malleolo-1},\\
x_3(t)={}&\sin (a \,t) \cos \left(\sqrt{1+a^2} \,t\right)\nn\\
&-\frac{a
   \sin \left(\sqrt{1+a^2}\, t\right) \cos (a
   \,t)}{\sqrt{1+a^2}}\label{malleolo-3}.
\end{align}

The corresponding  controls can be obtained via formulas  \eqref{sr-1}, \eqref{sr-2}, and \eqref{etmoide}, or
from equation~\eqref{uuu}, providing
\begin{align}
u_1(t)&=-\dot x_1(t)/x_2(t)=\pm\cos(a t),\nn\\
u_2(t)&=\dot x_3(t)/x_2(t)=\mp\sin(a t).\nn
\end{align}

\noindent\textbf{Step 3.} In this step, we have to find the initial covector (i.e., $p_\theta(0)\in\{-1,+1\}$ and $a=p_\varphi\in\R$) whose corresponding trajectory arrives at the final target $(x_1,x_2,x_3)(T)=(0,0,1)$.

From expression~\eqref{malleolo-2}, requiring $x_2(T)=0$ we get $\sin(\sqrt{1+a^2} \,T)= 0$.
Then from \eqref{malleolo-1},  requiring $x_1(T)=0$ we arrive at $\cos(a\,T)=0$.  We have then the conditions
\[
\sqrt{1+a^2} T=n_1 \pi,\quad a T=\frac\pi2+n_2 \pi,
\]
with $n_1,n_2\in\mathbb{Z}$.
Notice that $n_1>0$ since
$T>0$, and hence
$$
\frac{a}{\sqrt{1+a^2}}=\frac{n_2+\frac12}{n_1},
$$
from which we deduce
that $|n_2+\frac12|<n_1$.
It follows that
\begin{equation}
a=\frac{(n_2+\frac12)/n_1}{\sqrt{1-(\frac{n_2+\frac12}{n_1})^2}}.
\label{sfenoide}
\end{equation}
The target is reached at time
\begin{equation}
T=\pi n_1\sqrt{1-\left(\frac{n_2+\frac12}{n_1}\right)^2}.
\end{equation}

\noindent\textbf{Step 4.} The previous step provided a discrete set of trajectories reaching the target. The ones arriving in shorter time corresponds to $n_1=1$ and $n_2=0$  or $n_2=-1$, for which
$T=\pi \frac{\sqrt{3}}{2}$ and $a=\pm1/\sqrt{3}$. The final expression of the optimal trajectories and optimal controls are
\[
\begin{cases}
x_1(t)=\cos^3(t/\sqrt{3})\\
x_2(t)=\pm \frac{\sqrt{3}}{2}\sin(2 t/\sqrt{3})\\
x_3(t)=-\eps\sin^3(t/\sqrt{3}),
\end{cases}
\]
and
\[
\begin{cases}
u_1(t)=\pm\cos(t/\sqrt{3})\\
u_2(t)=\mp\eps\sin(t/\sqrt{3}),
\end{cases}
\]
where $\eps=\pm1$ is the sign of $a$.

We point out that the two trajectories arriving at time $T$ at the point $(0,0,1)$ are those corresponding to $n_2=0$ (i.e., to $a=1/\sqrt{3}$) and to $p_\theta(0)=\mp1$ (i.e., to the sign $+$ in formula \eqref{malleolo-2}), while the two  trajectories arriving at time $T$ at the point $(0,0,-1)$ are associated with $n_2=-1$ (i.e., to $a=-1/\sqrt{3}$).
See Section~\ref{s:excexc1}.
The four trajectories obtained in this way are optimal since we know that optimal trajectories exist and the ones that we have selected are the best among all trajectories satisfying a necessary condition for optimality. In this example, note that we can select
the optimal trajectories by applying the PMP and finding by hand the best extremals, without using second order conditions nor other sufficient conditions for optimality.\\

Figures~\ref{figgru1} and \ref{figgru2} display, respectively, the two symmetric extremal trajectories reaching the target state $(0,0,1)$ at time $T$ and the time evolution of the corresponding controls.
\begin{figure}
\begin{center}
\includegraphics[scale=0.5]{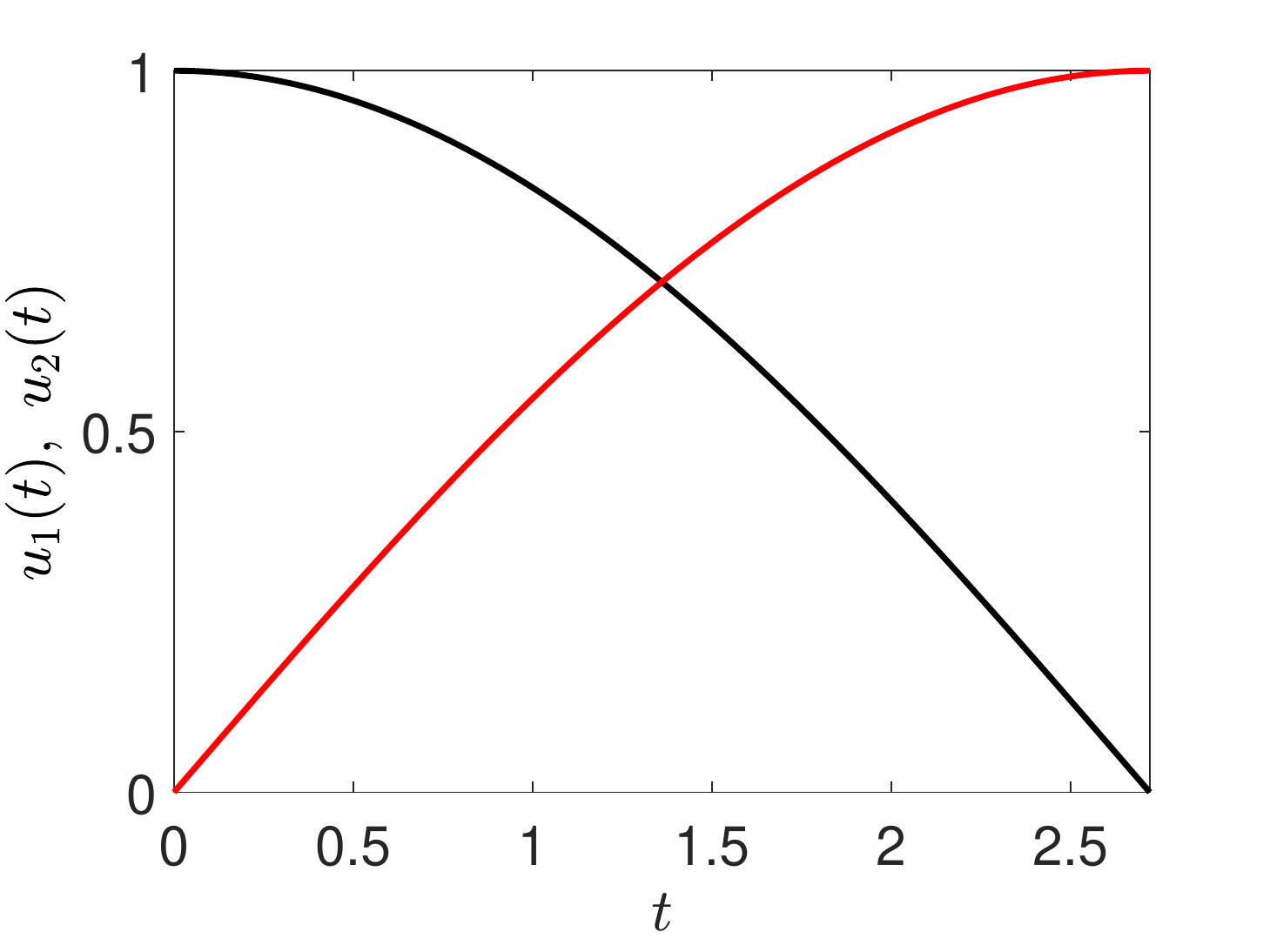}
\end{center}
\caption{(Color online) Plot of the two extremal trajectories (in red and blue) on the sphere going from the point $(1,0,0)$ to the point $(0,0,1)$ and minimizing the cost functional $C$.\label{figgru1}}
\end{figure}
\begin{figure}
\begin{center}
\includegraphics[scale=0.5]{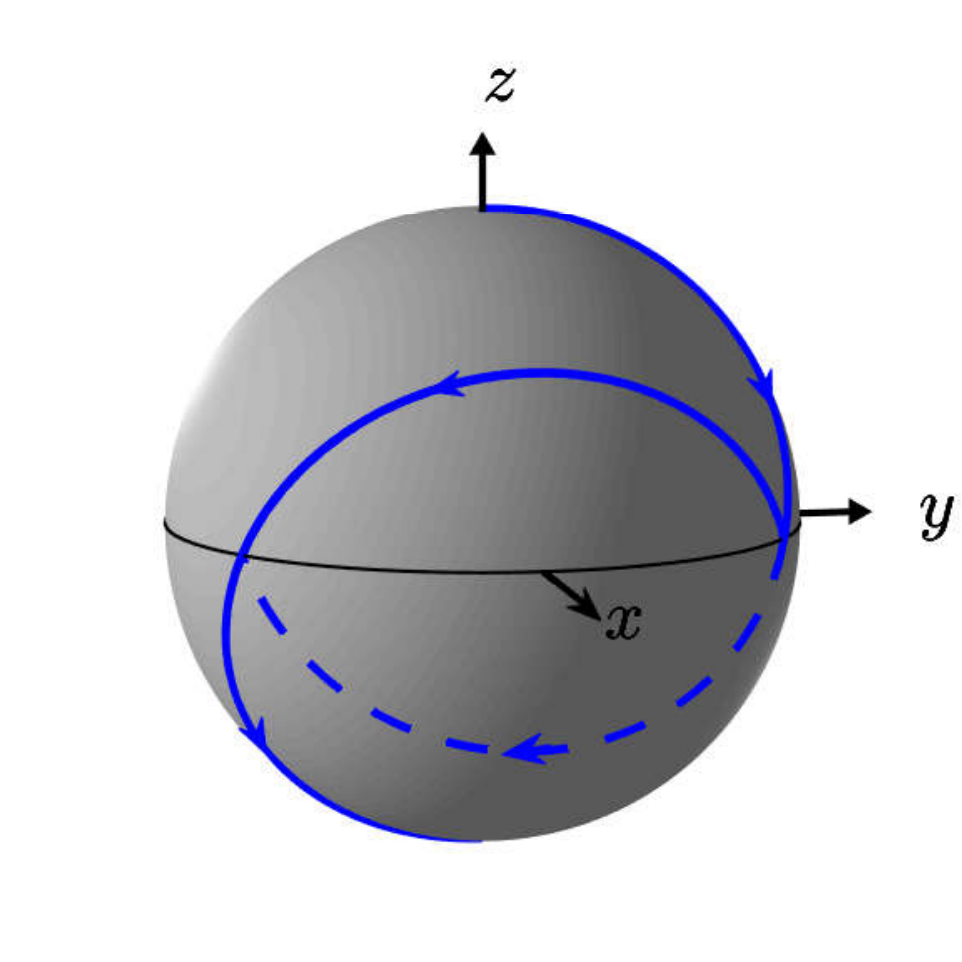}
\end{center}
\caption{(Color online) Time evolution of the controls $u_1$ (in black) and $u_2$ (in red). The control time $T$ and the parameter $a$ are respectively set to $\pi\sqrt{3}/2$ and $-1/\sqrt{3}$.\label{figgru2}}
\end{figure}

\section{Example 2:  A
minimum time
two-level quantum system with a real control}\label{sectwo}
\subsection{Formulation of the control problem}
We consider the control of a spin-1/2 particle whose dynamics are governed by the Bloch equation in a given rotating frame~\cite{ernstbook,levittbook}:
$$
\begin{cases}
\dot{M}_x=-\omega M_y, \\
\dot{M}_y=\omega M_x-\omega_x(t) M_z, \\
\dot{M}_z=\omega_x(t) M_y,
\end{cases}
$$
where $\textbf{M}=(M_x,M_y,M_z)$ is the magnetization vector and $\omega$ the offset term. The system is controlled through a single magnetic field along the $x$ axis that satisfies the constraint $|\omega_x|\leq \omega_{\max}$.  We introduce normalized coordinates $(x,y,z)= \textbf{M}/M_0$, where $M_0$ is the thermal equilibrium magnetization, a normalized control $u =\omega_x/\omega_{\max}$ which satisfies the constraint $|u|\leq 1$, and a normalized time $\tau= \omega_{\max}t$ (denoted $t$ below). Dividing the previous system by $\omega_{\max}M_0$, we get that the time evolution of the normalized coordinates is given by the  equations
$$
\begin{pmatrix}
\dot{x} \\
\dot{y} \\
\dot{z}
\end{pmatrix}
=
\begin{pmatrix}
-\Delta y \\
\Delta x \\
0
\end{pmatrix}
+u(t)
\begin{pmatrix}
0 \\
-z \\
y
\end{pmatrix},
$$
where $\Delta=\frac{\omega}{\omega_{\max}}$ is the normalized offset. The trajectories of the system lie on the Bloch sphere defined by the equation $x^2+y^2+z^2=1$. The manifold $M$ is here the sphere $S^2$. The differential system can be written in a more compact form as
\begin{equation}\label{eqspinq}
\dot{q}=(\mathcal{F}+u\mathcal{G})q,
\end{equation}
where $q=(x,y,z)$ is the state of the system, $u(t)\in U=[-1,1]$,  and $\mathcal{F}$ and $\mathcal{G}$ are the
skew-symmetric $3\times 3$ matrices
$$
\mathcal{F}=\begin{pmatrix}
0 & -\Delta & 0\\
\Delta & 0 & 0\\
0 & 0 & 0
\end{pmatrix}
,\qquad \mathcal{G}=\begin{pmatrix}
0 & 0 & 0\\
0 & 0 & -1\\
0 & 1 & 0
\end{pmatrix}.
$$
The vector fields  $\mathcal{F}q$ and $\mathcal{G}q$ generate rotations around, respectively, the $z$ and the $x$ axis.\\

\noindent\textbf{Existence of time-optimal trajectories.} By Point 1 in Proposition~\ref{p-controllability},
any initial point $q_0$ on the Bloch sphere $M$ can be connected by an admissible trajectory of the control system to any other point $q_1$ on the Bloch sphere (see also \cite{schirmer:2001}). In this case the existence of a time-optimal trajectory connecting $q_0$ to $q_1$ is a direct consequence of Proposition~\ref{p-tempomin} or Proposition~\ref{p-teschio-quantico}. Indeed, $U=[-1,1]$ is compact and the set $\{\mathcal{F}q+u\mathcal{G}q \mid u\in U\}$ is convex for any point $q$ of the Bloch sphere.

\subsection{Application of the PMP}\label{appPMPspin}
The goal of the two control processes that we are considering is to steer the system in minimum time from the north
pole $(0,0,1)$
of the Bloch sphere to, respectively,
the south pole $(0,0,-1)$  or the state $(1,0,0)$. We follow here the results established in \cite{boscain-mason,boscainchitour} (see also \cite{sugny10} for an experimental implementation).

The time-optimal control problem is solved by the application of the PMP. Here the cost $C$ to minimize is
$$
C=\int_0^Tdt=T\to\textrm{min},
$$
where $T$ is free. The pre-Hamiltonian
can be expressed as
$$
\HHH(q,p,u,p_0)=p(\mathcal{F}+u\mathcal{G})q+p^0,
$$
where $p=(p_x,p_y,p_z)\in\mathbb{R}^3$ is the covector
and $p^0$ is a nonpositive constant such that $p$ and $p^0$ are not simultaneously equal to 0. (Here $p$ is seen as a row vector and $q$ as a column vector.) The
value of ${\cal H}$ is constantly equal to $0$ since the final time is free.
The PMP states that the optimal trajectories are solutions of the equations
\begin{eqnarray*}
& & \dot{q}(t)=\frac{\partial \HHH}{\partial p}(q(t),p(t),u(t),p^0),\\
& & \dot{p}(t)=-\frac{\partial \HHH}{\partial q}(q(t),p(t),u(t),p^0), \\
& & \HHH(q(t),p(t),u(t),p^0)=\max_{|v|\leq 1}\HHH(q(t),p(t),v,p^0).
\end{eqnarray*}
The dynamics of the adjoint state $p$ are given by
\begin{equation}\label{sansdag}
\dot{p}(t)=
-p(t)(\mathcal{F}+u(t)\mathcal{G}).
\end{equation}
 Note that $\frac{d}{dt}\|p(t)\|^2=0$, so that $\|p(t)\|$ does not depend on time.
 Its constant value is nonzero since $\HHH=0$ and $(p,p_0)\ne 0$.\\

\noindent\textbf{Steps 1 and 2.}
 Since the only term of  the pre-Hamiltonian depending on the control is
$up\mathcal{G}q$,
the maximization condition of the PMP leads to the introduction of the switching function
$$
\Phi(t)=p(t) \mathcal{G}q(t).
$$
In the regular case in which $\Phi(t)\neq0$, we deduce from the maximization condition that the optimal control is given by the sign of $\Phi$, $u(t)=\textrm{sign}[\Phi(t)]$.
The corresponding trajectory is called a \emph{bang trajectory}. If $\Phi$ has an isolated zero in a given time interval, then the control
function may \emph{switch} from $-1$ to $1$ or from $1$ to $-1$. A \emph{bang-bang trajectory} is a trajectory obtained after a finite number of switches.

\noindent Using the relation~\eqref{sansdag},
we have
\begin{equation}\label{eq:derivative-switching}
\dot{\Phi}(t)=p(t) [\mathcal{G},\mathcal{F}]q(t),
\end{equation}
where $[\cdot,\cdot]$ denotes the matrix commutator operator. In particular, $\Phi$ is a $C^1$ function and, for almost every $t$,
\begin{equation*}
\ddot{\Phi}(t)=
p(t) [\mathcal{F},[\mathcal{F},\mathcal{G}]]q(t)+u(t)p(t) [\mathcal{G},[\mathcal{F},\mathcal{G}]]q(t).
\end{equation*}
Since $[\mathcal{F},[\mathcal{F},\mathcal{G}]]=-\Delta^2 \mathcal{G}$ and $[\mathcal{G},[\mathcal{F},\mathcal{G}]]=\mathcal{F}$,
we have,  for a.e. $t$,
\begin{align}
\ddot{\Phi}(t)&=
-\Delta^2 \Phi(t)+u(t) p(t) \mathcal{F}q(t)\nn\\
&=-\Omega^2
\Phi(t)-p^0 u(t),\label{eq:second-derivative-switching-rivista}
\end{align}
where $\Omega=\sqrt{1+\Delta^2}$ and the second equality follows from the identity $\HHH=0$.

\noindent\textbf{Abnormal extremals.} Abnormal extremals are characterized by the equality $p^0=0$,
from which, together with $\HHH=0$, we deduce that
$\Phi(t)=-p(t)\mathcal{F}q(t)/u(t)$. If $\Phi(t)=0$ at time $t$, then $p(t)$ is orthogonal both to $\mathcal{F}q(t)$ and $\mathcal{G}q(t)$.
If, moreover, $t$ were not an isolated zero of $\Phi$, then
$\dot\Phi(t)=0$, since $\Phi$ is $C^1$. It would follow from
\eqref{eq:derivative-switching} that $p(t)$ is orthogonal also to
$[\mathcal{G},\mathcal{F}]q(t)$. Since
for every $q\in M$ the vectors
$\mathcal{F}q$, $\mathcal{G}q$, and $[\mathcal{G},\mathcal{F}]q$ span the tangent plane $T_q M$ to the sphere $M$,
  we would deduce that $p(t)=0$, contradicting the PMP. This means that abnormal extremals are necessarily bang-bang.
Moreover, we deduce from \eqref{eq:second-derivative-switching-rivista} that
the switching times are the zeros of a nontrivial solution of the equation
\[\ddot\Phi+\Omega^2\Phi=0.\]
The length of an arc between any two successive switching times is then equal to
 $\pi/\Omega$.

\noindent\textbf{Singular arcs.} When the trajectory is normal, there might exist
extremals for which $\Phi$ is zero on a nontrivial time interval. The
control is singular on such an interval, since it cannot directly be obtained from the maximization condition.
The restriction of the trajectory to an interval on which $\Phi\equiv 0$ is called a \emph{singular arc}.
Singular arcs are characterized by the fact that the  time derivatives of $\Phi$ at all orders are zero.
Since $p(t)$ is different from zero, the only possibility to have simultaneously $\Phi(t)=0$ and $\dot{\Phi}(t)=0$ is that the vectors $\mathcal{G}q(t)$ and $[\mathcal{G},\mathcal{F}]q(t)$ are parallel. Since $\mathcal{G}q(t)$ and $[\mathcal{G},\mathcal{F}]q(t)$ generate, respectively, the rotations around the $x$ and the $y$ axis, we deduce that singular arcs are contained in the equator $z=0$ of the sphere. The singular control law
 $u_s$ can be calculated from \eqref{eq:second-derivative-switching-rivista} by imposing that $\Phi$ and its second time derivative are zero, yielding
 $u_s(t)=0$.
As it could be expected, this control law
generates a rotation along the equator. It is admissible because $|u_s|\leq 1$.

\noindent\textbf{Normal bang-bang extremals.} Consider a normal extremal and
an
interior bang arc of duration $T$ between the switching times $t_0$ and $t_0+T$ on which the control $u$ is constantly equal to $+1$ or $-1$.  An  interior
arc is an arc which is neither at the beginning nor at the end of the extremal.
We normalize $p^0$ to $-1$.
According to \eqref{eq:second-derivative-switching-rivista}, the function $\Psi=\Phi-\frac{u}{\Omega^2}$ is
a solution of $\ddot \Psi+\Omega^2 \Psi=0$.
Moreover, since $\Phi$ is non-constant then $\Psi$ is nontrivial.
Hence, $\Psi(t)=\nu \cos(\Omega t+\theta_0)$ for some $\nu>0$ and $\theta_0\in \mathbb{R}$.
Moreover, $\nu>0$ is uniquely identified by $u$ and $\dot\Phi(t_0)$ through the equalities $\Psi(t_0)=-\frac{u}{\Omega^2}$ and $\dot \Psi(t_0)=\dot \Phi(t_0)$.
Switchings occur if $\Phi=\Psi+\frac{u}{\Omega^2}$ vanishes and changes sign.
Since, moreover, $\mathrm{sign}[u]=\mathrm{sign}[\Phi]$,
it follows that $\Psi$ is larger than the negative value $-\frac{1}{\Omega^2}$ on $(t_0,t_0+T)$ when $u=+1$ and smaller
than the positive value $\frac{1}{\Omega^2}$ on $(t_0,t_0+T)$ when $u=-1$.
Hence $T$ is larger than $\pi/\Omega$ and
\[\dot \Phi(t_0+T)= -\dot \Phi(t_0).\]
If $\dot \Phi(t_0)=0$, we deduce that $T=2\pi/\Omega$.
Notice that if $u$ is constantly equal to $+1$ or $-1$, then the
solutions of $\dot q=(\mathcal{F}+u\mathcal{G})q$
are $2\pi/\Omega$ periodic rotations around the axis spanned by $(u,0,\Delta)$. Since a time-optimal trajectory cannot self-intersect, we conclude that
$T<2\pi/\Omega$ and $\dot \Phi(t_0)\ne 0$.

If $t_0+T$ is the starting time of another internal bang arc, then by the above considerations the duration of such internal bang arc is also equal to $T$. Given a normal bang-bang trajectory, there exists then
 $T\in (\pi/\Omega,2\pi/\Omega)$ such that the trajectory
is the concatenation of bang arcs of duration $T$, except possibly for the first and last bang arc, whose length can be smaller than $T$.

\noindent\textbf{General extremals.} As we have seen in the previous paragraphs, if a trajectory contains an internal bang arc, then it is bang-bang.
Otherwise the set of zeros of $\Phi$ is connected, that is, either $\Phi$ has a single zero or it vanishes on a nontrivial singular arc and is different from zero out of it.

To summarize, extremal trajectories are of two types:
\begin{itemize}
\item Bang-bang trajectories whose internal bang arcs have all the same length $T\in [\pi/\Omega,2\pi/\Omega)$ (the case $T=\pi/\Omega$ corresponding to abnormal extremals) and for which the first and last bang arcs have length at most $T$.
\item Concatenations of a (possibly trivial) bang arc of length smaller that $2\pi/\Omega$, a singular arc on which $u_s=0$, and another (possibly trivial) bang arc of length smaller than $2\pi/\Omega$.
\end{itemize}

 \noindent\textbf{Steps 3 and 4.} We solve in this paragraph two time-optimal control problems. Starting from the north pole $(0,0,1)$, the goal is to reach
in minimum time the points $(0,0,-1)$ (problem (P1)) and $(1,0,0)$ (problem (P2)). To simplify the derivation of the optimal solutions, we assume that $|\Delta|\leq 1$~\cite{boscain-mason}.

Before solving (P1) and (P2), we first derive analytical results describing the dynamics of the system. Consider a bang extremal trajectory starting from the north pole at time $t=0$ with control $u(t)=\varepsilon=\pm 1$. The corresponding trajectory is given by
$$
\begin{cases}
x(t)=\frac{\varepsilon \Delta}{\Omega^2}(1-\cos(\Omega t)) \\
y(t)=-\frac{\varepsilon}{\Omega}\sin(\Omega t)\\
z(t)=1+\frac{1}{\Omega^2}(\cos(\Omega t)-1).
\end{cases}
$$
The first two times for which $z(t)=0$ are $t_1=\frac{1}{\Omega}(\pi-\arccos(\Delta^2))$ and $t_2=\frac{1}{\Omega}(\pi+\arccos(\Delta^2))$.
Notice that all other times for which $z(t)=0$ are larger than $2\pi/\Omega$ and cannot be the duration of a bang arc of an optimal trajectory.

The optimal solution of (P1) is a bang-bang trajectory with a first switch on the equator at $t=t_1$. The total duration of the process is $t_1+t_2=\frac{2\pi}{\Omega}$. A symmetric configuration is possible with a first switch at $t=t_2$. The two trajectories are displayed in Fig.~\ref{figspin2}.
Let us discuss how the optimality of such a trajectory can be asserted.
The proposed trajectory is clearly extremal and connects the chosen initial and final points. Since a bang-bang trajectory with at least two internal bang arcs has duration larger than $2\pi/\Omega$, it follows that any bang-bang trajectory with four or more bang arcs has a duration larger than the candidate optimal trajectory (that is, $\frac{2\pi}{\Omega}$). One is then left to compare $\frac{2\pi}{\Omega}$
with finitely many types of trajectories: bang-bang trajectories with two or three bang arcs, and trajectories obtained by concatenation
of a bang, a singular, and a bang arc. By setting the initial and final points, this leaves  few competitors to the optimal trajectory, which can be easily excluded by enumeration. As an example, we consider the concatenation of a bang arc with $u=-1$, a singular arc and a new bang arc with $u=+1$. The two bang arcs
last for a time $t_1$, while the duration of the singular arc is $\frac{2}{\Delta}\arctan(\frac{\sqrt{1-\Delta^2}}{\Delta})$. It is then straightforward to deduce that the duration of this extremal solution is larger than $t_1+t_2=\frac{2\pi}{\Omega}$, except in the case $\Delta=1$ for which the singular arc is of length zero.

\begin{figure}
\begin{center}
\includegraphics[scale=0.8]{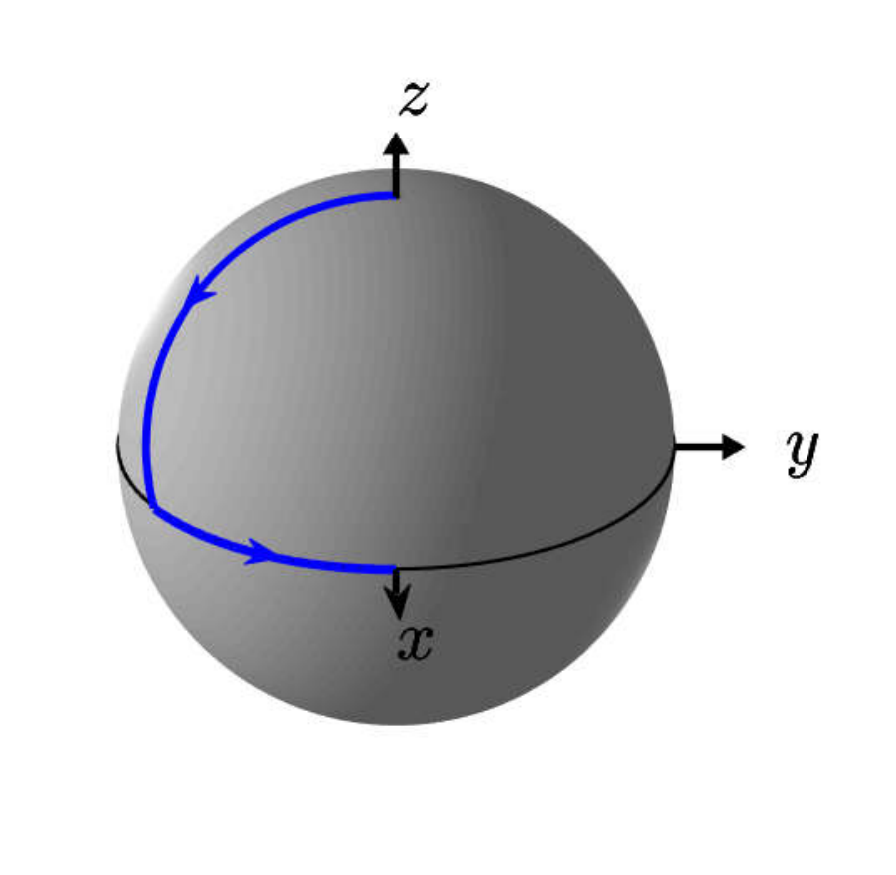}
\end{center}
\caption{(Color online) Optimal trajectories (in blue) going from the north pole to the south pole of the Bloch sphere. The solid black line indicates the position of the equator. The parameter $\Delta$ is set to $-0.5$.
 For one of the two optimal trajectories
a control $u=-1$ is first applied during a time $t_1$,
followed by a control $u=+1$ during a time $t_2$, while for the other trajectory $u=-1$ lasts for a time $t_2$ and $u=+1$ for a time $t_1$.\label{figspin2}}
\end{figure}

Let us now discuss the  solution of (P2).
Using the results of Sec.~\ref{appPMPspin}, we consider the concatenation of a bang extremal with $u=+1$ during the time $t_1$ and of a singular extremal with $u=0$ during the time $t_s$. At time $t=t_1$, the trajectory reaches the point $(\Delta,-\sqrt{1-\Delta^2},0)$. We deduce that $t_s=\frac{1}{\Delta}\arctan(\frac{\sqrt{1-\Delta^2}}{\Delta})$. The total duration of the control process is $t_1+t_s$. The corresponding
trajectory is represented in Fig.~\ref{figspin1}.
One can check that all other candidates for optimality
join $(1,0,0)$ in a longer time.
The situation is more complicated than for (P1), since here $t_s\to \infty$ as $\Delta\to 0$, so the candidate trajectory should be compared with trajectories with more and more bangs as $\Delta\to 0$.
A proof of the optimality of the trajectory described above can be obtained, for instance, using optimal synthesis theory, i.e., describing all the optimal trajectories starting from the north pole, as done in~\cite{boscain-mason}.

\begin{figure}
\begin{center}
\includegraphics[scale=0.9]{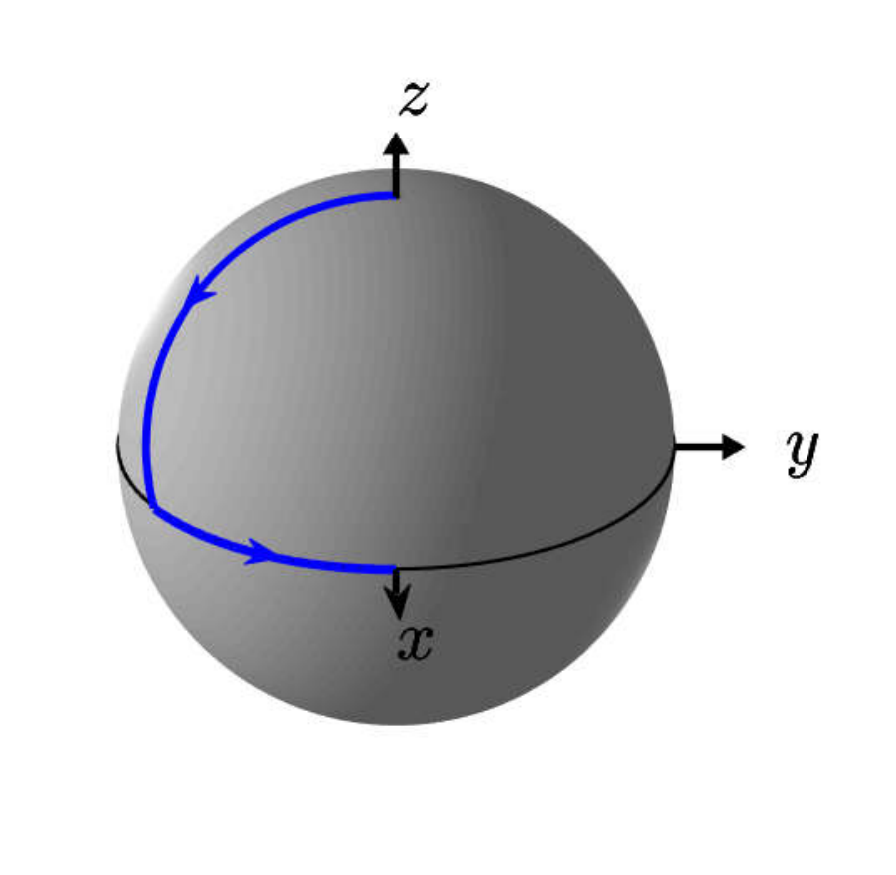}
\end{center}
\caption{(Color online) Optimal trajectory (in blue) going from the north pole to the point $(1,0,0)$ of the $x$ axis. The solid black line indicates the position of the equator. The parameter $\Delta$ is set to 0.5.\label{figspin1}}
\end{figure}

\section{Applications of quantum optimal control theory}\label{applioct}
We discuss in this section the link between the material and results presented in this tutorial and the current research objectives in this field, both from a theoretical and an experimental points of view. The ability to quickly and accurately perform operations in a quantum device is a key task in quantum technologies. Quantum optimal control theory addresses this challenge by combining analytical and numerical tools to design procedures adapted to the experimental setup under study. This approach has the key advantage of being based on a rigorous theoretical framework on which all developments have been based since the 1980s. In the case of a low-dimensional quantum system as the two examples presented in Sec.~\ref{secthree} and \ref{sectwo}, the optimal control problem may be solved analytically or at least with a very high numerical accuracy. For high-dimensional systems, an efficient alternative is provided by numerical optimal control procedures, which include first and second-order gradient ascent algorithms~\cite{KHANEJA2005,machnes2011,fouquieres2011,ciaramella2015} whose structure is based on the PMP. The connection between the PMP and gradient-based optimization algorithms is described in Sec.~\ref{sectiongradient}. Their flexibility allows them to adapt to many experimental situations for which precise modeling of the dynamics is known. As concrete and recent examples, we mention the experimental implementation of such techniques in Rydberg atoms~\cite{omran2019,koch20}, in spin-wave states of atomic ensembles~\cite{he2020}, in electron spin resonance~\cite{Probst2019} or in superconducting cavity resonator~\cite{heeres2017}. The diversity of these examples shows the key role that optimal control now plays in quantum technologies. We describe below some examples of recent applications which are based on the mathematical formalism described in this tutorial. We also indicate some open issues in this field. We stress that our aim is not to provide an exhaustive list of all optimal control studies in this very active area (we refer the interested reader to recent reviews on this subject~\cite{cat,altafini-ticozzi,dong,past-present-future}), but rather to give an overview of the different aspects that can be analyzed.

The analytical techniques have been applied recently to a series of fundamental and practical issues in quantum technologies such as robustness~\cite{dridi20,vandamme17,barnes1,barnes2,barnes3,li2017,martikyan20} and selectivity~\cite{vandamme18,haidong20} of the control process with respect to parameters of the system Hamiltonian. The robustness property is a key factor for the experimental implementation of open-loop optimal pulses. The basic idea consists in controlling simultaneously an ensemble of identical systems which differ only by the value of the unknown parameter. Robust or selective optimal pulses are obtained respectively in the case where the target states are the same or are different according to the systems. In the same direction, numerical optimal control has been also employed to improve the estimation of Hamiltonian parameters~\cite{haidong2015,ansel2017}. Here, from a specific cost functional, the control is used to optimize the discrimination between the system dynamics, which leads to a precise characterization of the Hamiltonian~\cite{hou2019}. Using the same type of approach, optimal control has proven to be very effective for quantum sensing~\cite{poggiali2018,calarcoreview,Konzelmann2018}.

A key aspect of quantum technologies is the minimum time to perform a given operation. Lower bounds on the time can be established in the framework
of quantum speed limits~\cite{Deffner_2017,frey} where the time is expressed as a ratio between the distance to the target state
and the dynamical speed of evolution. Time-optimal trajectories correspond by definition to the speed limit of the processes. The solutions of the two examples solved in this tutorial give therefore the quantum speed limits of the problems under study. However, in a higher dimensional quantum system, the many local minima of the control landscape make it very difficult to find a good approximation of the optimal trajectory and lead generally to an upper bound of the minimum time. Recent studies have explored the link between these two formalisms~\cite{hegerfeldt2013,caneva}. A major step forward would be to rigorously describe the concept of speed limit in terms of the Pontryagin Maximum Principle~\cite{gajdacz2015,diaz20}.

Quantum thermodynamics is nowadays a very useful tool to describe the dynamics of open quantum systems used in quantum technologies~\cite{kosloff2013}. Optimal control and the PMP can be applied to analyze many issues in this field such as the enhancement of quantum engine performances~\cite{cavina2017,cavina2018}, work extraction~\cite{manzano2018,miller2019} or qubit purification~\cite{basilewitsch2017,fischer19,basilewitsch2020}. The quantification of the ressources to control a quantum system has been also discussed and could be optimized from the PMP~\cite{abah2019}.

An intense effort is being made to develop new quantum optimal control algorithms, which are well suited to experimental limitations and constraints~\cite{egger2014,wilhem2020,mintert2020,mintert2020_2,dridi15} or which are particularly effective in designing optimal pulses~\cite{machnes2018,kliclstein2017,jun2017}. In particular, hybrid optimal control using feedback from the experiment has been developed to avoid problems due to inaccuracies in the modeling of the dynamics. The structure of the control landscape for the preparation of quantum many-body systems has been identified and exhibits a spin-glass-like phase transition in the space of control protocols close to the optimal fidelity~\cite{day}. This structure explains the difficulty to find numerically the optimal solution in this region. Combining optimal control and machine learning approaches as proposed in~\cite{bukov2018} could be a way to solve this problem and to improve the efficiency of optimal algorithms~\cite{niu2019,perrier2020}. However, we stress that a rigorous mathematical description of these numerical methods and results will be necessary in the future to systematically apply these approaches to other quantum devices.

We conclude this section by pointing out that the PMP may also play a more unexpected role in quantum computing. It has recently been shown that Grover's quantum search problem can be mapped to a time-optimal control problem, and then described through the PMP~\cite{lin2019}. A connection between optimal control theory and variational quantum algorithms~\cite{VQA} has been established in~\cite{yang2017}. A general description of this intrinsic link is given in~\cite{rabitz2020} where optimal control is used to precisely adjust the parameters of a quantum circuit. These results highlight that the PMP as a general optimization tool is not only interesting for the computation of time-dependent control pulses, but also in other optimization problems of interest in quantum technologies.

\section{Conclusion}\label{conclu}
In this tutorial, we have attempted to give the reader a minimal background on the mathematical techniques of OCT.
In our opinion,  this
is a fundamental prerequisite to rigorously and correctly apply these tools in quantum control.

The objectives of the tutorial are twofold. First, we have highlighted the key concepts of the PMP using ideas based on the finding of extrema of functions of several variables. This analogy gives non-experts an intuition of the tools of optimal control that might seem abstract on first reading. We have then stated the PMP and described in details the different steps to be followed in order to solve an optimal control problem. Some are rarely discussed in quantum control, such as the existence of solutions or abnormal and singular extremals, while they play a crucial role in some problems. The link between the PMP and gradient-based optimization algorithms has been explained, which also highlights the role of such mathematical tools in any numerical optimization process. Second, we have solved two basics control problems, namely the control of a three-level quantum system by means of two complex resonant fields and the control of a spin 1/2 particle through a real off-resonance driving. The low dimension of the two systems allows us to express analytically the optimal solutions and to give a complete geometric description of the control protocol. Such examples can be used as a starting point by the reader to apply the PMP to more complex quantum systems. A series of recent results using OCT has been briefly discussed to highlight current research directions in this area.

\section{List of notations}\label{notations}
We provide in Tab.~\ref{tabres} a list of the main notations used in the paper as well as the first section where  they appear or are defined.

\begin{table}[tb]
\caption{Main notations used in the paper.\label{tabres}}
\begin{tabular}{|c|c|c|}
\hline
{\bf Notation}& Definition & Section \\
\hline\hline
$q(t)$ & the state of the system at time $t$ & Sec.~\ref{secformu} \\
\hline
$f$ & a smooth vector-valued function of $q$ and $u$ which defines the control system & Sec.~\ref{secformu} \\
\hline
$f_0$ & a smooth scalar-valued function of $q$ and $u$ which defines the cost functional & Sec.~\ref{secformu} \\
\hline
$M$ & the manifold on which the state $q(t)$ evolves & Sec.~\ref{secformu} \\
\hline
$T$ & the final control time & Sec.~\ref{secformu} \\
\hline
$u(t)$ & the value of the control law at time $t$ & Sec.~\ref{secformu} \\
\hline
$U$ & the set of possible values of $u(t)$ & Sec.~\ref{secformu} \\
\hline
$\mathcal{U}$ & the set of admissible control laws & Sec.~\ref{secformu} \\
\hline
$\mathcal{T}$ & the set of target states & Sec.~\ref{secformu} \\
\hline
$d(\mathcal{T},q(T))$ & the distance between $\mathcal{T}$ and the final state $q(T)$
& Sec.~\ref{secformu} \\
\hline
$\mathcal{R}(q_{\textrm{in}})$ & the set of reachable states from $q_{\textrm{in}}$ & Sec.~\ref{secformu} \\
\hline
$\mathcal{R}^T(q_{\textrm{in}})$ & the set of reachable states in time $T$ from $q_{\textrm{in}}$ & Sec.~\ref{secformu} \\
\hline
$\textbf{F}(q)$ & the set of directions $f(q,u)$ at $q$ as $u$ varies in $U$ & Sec.~\ref{s-existence} \\
\hline
$\phi(q(T))$ & the terminal cost which depends on the final state $q(T)$ & Sec.~\ref{s-existence} \\
\hline
$J(u(\cdot))$ & the value of the cost for a control $u(\cdot)$ & Sec.~\ref{s-first} \\
\hline
$\mathcal{H}$ & the pre-Hamiltonian in the PMP & Sec.~\ref{s-first}\\
\hline
$p$ & the adjoint state in the PMP & Sec.~\ref{s-first}\\
\hline
$p^0$ & the abnormal multiplier & Sec.~\ref{s-first}\\
\hline
\hline
\end{tabular}
\end{table}

\appendix

\section{Test of controllability}\label{testcontrollability}
We give in this section a general and useful sufficient condition for controllability. This section completes the discussion of Sec.~\ref{secformu}.
\bp\label{p-controllability}
Let $M$ be a smooth manifold and $U$ be a subset of $\R^m$ containing a neighborhood of the origin.
Consider a control system
of the form $\dot q=F_0(q)+\sum_{j=1}^m u_j(t) F_j(q)$, with
$F_0,\dots,F_m$ smooth vector fields on $M$ and
$u(t)=(u_1(t),\dots,u_m(t))\in U$.
Let ${\cal L}_0$ be the Lie algebra generated by the vector fields $F_0,\dots,F_m$ and
${\cal L}_1$ be the Lie algebra generated by the vector fields $F_1,\dots,F_m$.
The system is controllable if at least one of the following conditions is satisfied:
\begin{enumerate}
\item $F_0$ is a recurrent vector field and $\dim({\cal L}_0(q))$ is equal to the dimension of $M$ at every $q\in M$;
\item $U=\R^m$ and  $\dim({\cal L}_1(q))$ is equal to the dimension of $M$ at every $q\in M$.
\end{enumerate}
\ep
For a definition of Lie algebra generated by a set of vector fields and for the notion of recurrent vector field we refer to \cite{agrachev-book,chapter}. For our purposes, it is sufficient to recall that:\\
- if $F_0(q)=A_0q,\dots,F_m(q)=A_mq$
are linear vector fields then $G\in {\cal L}_0$ if and only if $G(q)=Bq$ with $B$ in the matrix Lie algebra (for the commutator product) generated by $A_0,\dots,A_m$;\\
-  if a vector field $F$ is such that every solution of $\dot q=F(q)$ is periodic, then $F$ is recurrent.\\
With these sufficient conditions, we recover the standard controllability conditions introduced in quantum control for closed systems~\cite{albertini-dalessandro,schirmer:2001} and recalled in Sec.~\ref{secformu}.

\section{Filippov's theorem}\label{filippovtheorem}
We state here Filippov's theorem, which gives a sufficient condition
for the compactness of the reachable set~\cite{liberzon-book}. This section completes the discussion of Sec.~\ref{s-existence}.
\bt[Filippov]
\label{t-filippov}
 Fix $T>0$. Consider the control system  $\dot \qq(t)= f(\qq(t),u(t)),$ $\qq\in M$,  $u\in \mathcal{U}$, where $M$ is a $n$-dimensional manifold and  $U\subset \R^m$.
 Fix an initial condition $\qqin\in M$.
 Assume the following conditions:
\begin{itemize}
\item the set $U$ is compact,
\item the set {\bf F}$\displaystyle (\qq)=\Big\{f(
\qq ,u)\mid u\in U \Big\}$ is convex for every $\qq\in M$,
\item for every $u\in \mathcal{U}$, the solution of $\dot \qq(t)= f(\qq(t),u(t)),$ $q(0)=\qqin$, is defined on the whole interval $[0,T]$.
\end{itemize}
Then
the sets ${\cal R}^T(\qqin)$ and ${\cal R}^{\leq T}(\qqin)$ are compact.
\et
Here
${\cal R}^{\leq T}(\qqin)$ denotes the reachable set from $\qqin$ within time $T$, i.e.,
the set of points $\bar\qq$ in $M$ such that there exist $T'\in[0,T]$
 and an admissible trajectory  $q:[0,T']\to M$  such that  $q(0)=\qqin$ and $q(T')=\bar\qq$.

Note that the third hypothesis of Th.~\ref{t-filippov}
is automatically satisfied when $M$ is compact.

By applying Filippov's theorem to the augmented system for problem {\bf (P1)}
introduced in Sec.~\ref{s-existence}, one obtains:
\begin{prop}\label{p-teschio}
Fix $T>0$. Assume that
\begin{itemize}
\item  $\target$ is closed and ${\cal R}^T(\qqin)\cap\target\neq\emptyset$,
\item the set $U$ is compact,
\item  the set $\hat {\bf F}(\qq)=\left\{  \left(\ba{c}f^0(\qq,u)\\ f(q,u)\ea \right) \mid u\in U \right\}$ is convex for every $q\in M$,
\item   for every $u\in \mathcal{U}$ the
solution of $\dot {\qq}(t)=f(\qq(t),u(t))$, $q(0)=\qqin$, is defined on the whole interval $[0,T]$.
\end{itemize}
Then there exists a solution to  problem {\bf P1}.
\end{prop}

\begin{example}
{ The conclusion of
Proposition~\ref{p-teschio} may fail to hold if we drop the convexity assumption on the set {\bf F}$\displaystyle (\qq)$ of admissible velocities,
as illustrated by the following example.
Take $M=\R^2$, $U=\{-1,1\}$, and $f(q,u)=A_u q$ with $A_1=\left(\begin{smallmatrix}-1&1\\-1&0\end{smallmatrix}\right)$ and $A_{-1}=\left(\begin{smallmatrix}-1&-1\\1&0\end{smallmatrix}\right)$. We point out that the corresponding control system satisfies all the assumptions of Th.~\ref{t-filippov} except for the convexity of {\bf F}$\displaystyle (\qq)$, since the control can only take two discrete values $-1$ and $1$.
 Pick $\qqin=(1,0)$ and any $T>0$. Then $q_T:=e^{-T}\qqin$ is in the closure of ${\cal R}^{T}(\qqin)$, since we can end up arbitrarily close to $q_T$ at time $T$ by applying a control
that switches fast enough between $-1$ and $1$. On the other hand, if $\dot q(t)=f(q(t),u(t))$ at time $t$ and $q(t)$ has a nonzero vertical component, then $\frac{d}{dt}\|q(t)\|>-\|q(t)\|$. Since every trajectory of the control system starting at $\qqin$ necessarily leaves the horizontal axis, its final point  at time $T$ has norm larger than $e^{-T}$. Consequently $q_T\not\in    {\cal R}^{T}(\qqin)$. This proves that $ {\cal R}^{T}(\qqin)$ is not closed, hence the conclusion of Th.~\ref{t-filippov} does not hold. This result is illustrated in Fig.~\ref{figswitch} which shows numerically that the distance to the target state tends to zero when the number of switches goes to $+\infty$. This example also highlights possible numerical problems when the existence of the optimal solution is not verified. Here, an optimization algorithm cannot converge towards a well-defined control law.}
\begin{figure}
\begin{center}
\includegraphics[scale=0.5]{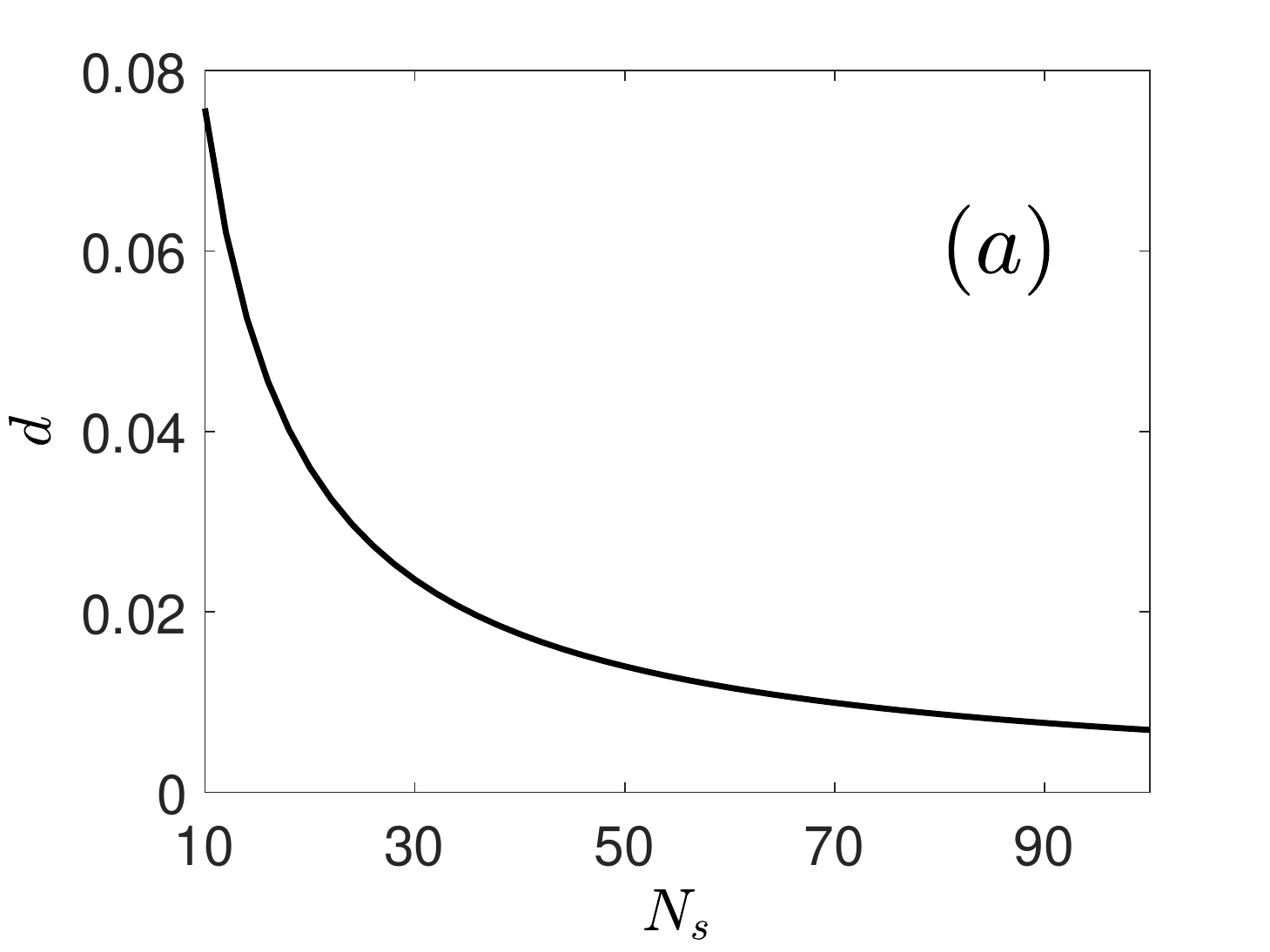}
\includegraphics[scale=0.5]{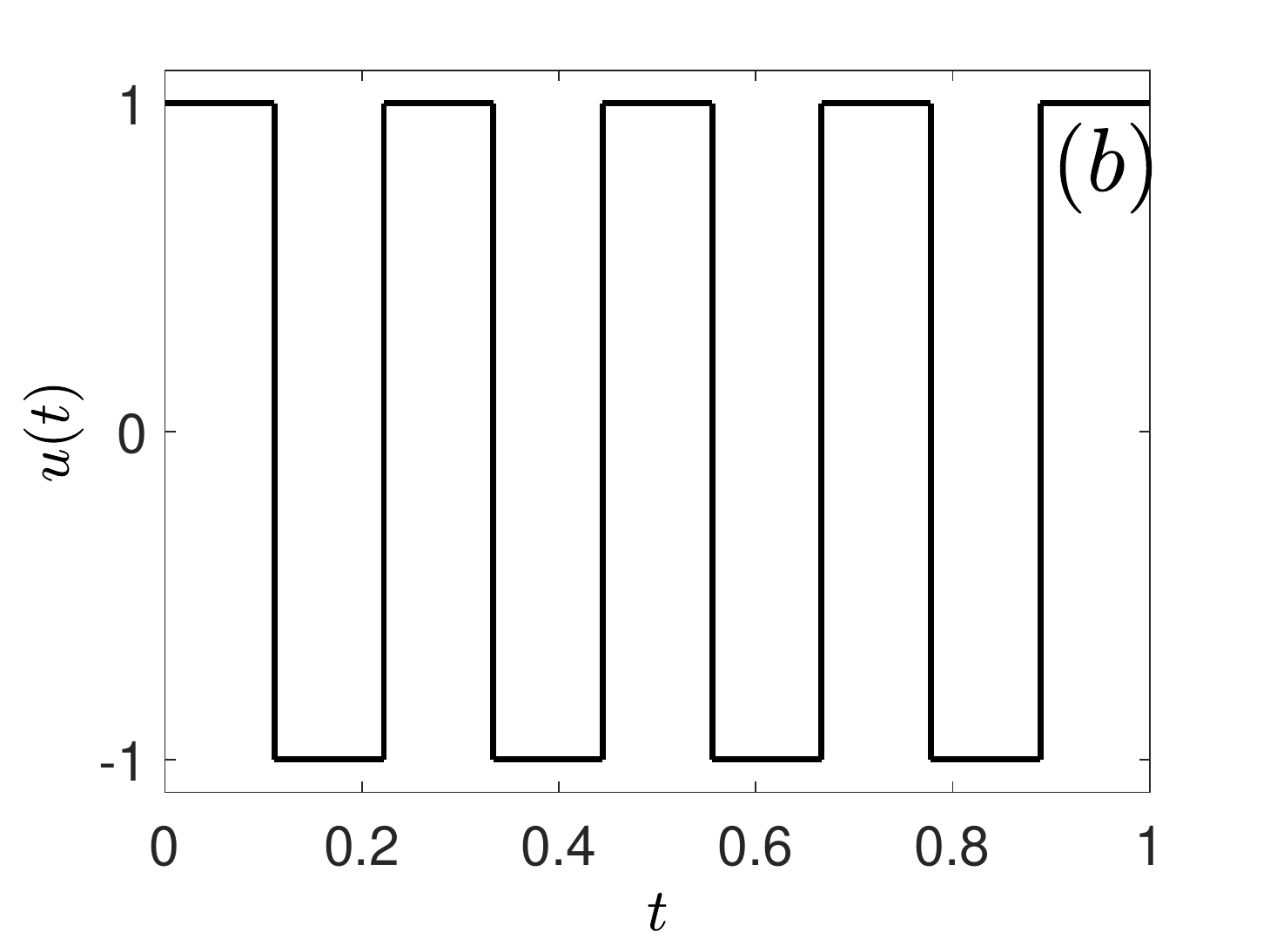}
\end{center}
\caption{Illustration of the existence problem. Panels (a) and (b) depict respectively the Euclidian distance $d$ to the target as a function of the number $N_s$ of switchings and an example of control law $u$ with $N_s=8$.\label{figswitch}}
\end{figure}
\end{example}

The idea of reducing the problem of existence of an optimal control to the compactness
of the reachable set of the augmented system can be used for more general problems.
For instance, if we add a terminal cost $\phi(\qq(T))$ to the cost $\int_0^T f^0(\qq(t),u(t))\,dt$, where $\phi$ is a smooth function (as for instance in Approach $B$, or in the general formulation given in Sec.~\ref{ss-PMP}), we get a similar result
adding to $f^0(\qq(t),u(t))$ the directional derivative of $\phi$ along $f(\qq(t),u(t))$, that is,
replacing $f^0(\qq(t),u(t))$ by $f^0(\qq(t),u(t))+\langle d\phi(\qq(t)),f(\qq(t),u(t))\rangle$.
We recall that here $d\phi(\qq(t))$ denotes the differential of
$\phi$ evaluated at $\qq(t)$ and that  $\langle\cdot,\cdot\rangle$
is the duality product between covectors of $T^{\ast}M$
and vectors of $TM$.

When the final time is free, it is more difficult to get the existence of optimal trajectories. However, the compactness of  ${\cal R}^{\leq T}(\qqin)$ in the Filippov theorem can be used to find conditions for the existence of optimal controls in minimum time.
 We state this result in the case where $M$ is compact. 
Note that  the problem of minimizing time can be written in the form of problem {\bf P1}
with $T$ free and $f^0=1$.
\bp\label{p-tempomin}
Consider problem {\bf P1} with $T$ free, $f^0=1$, and $M$ compact. Assume that
\begin{itemize}
\item  $\target$ is closed and ${\cal R}(\qqin)\cap\target\neq\emptyset$,
\item the set $U$ is compact,
\item  the set $ {\bf F}(\qq)=\{   f(q,u) \mid u\in U \}$ is convex for every $q\in M$.
\end{itemize}
Then there exists a solution to the problem.
\ep

A straightforward application of Propositions~\ref{p-teschio} and \ref{p-tempomin} to closed quantum systems yields Proposition~\ref{p-teschio-quantico} of the main text.

\subsection*{Acknowledgments}
The authors have contributed equally to this work. This research has been partially supported by the ANR projects ``SRGI'' ANR-15-CE40-0018 and ``QUACO'' ANR-17-CE40-0007-01.  This project has received funding from the European Union Horizon 2020 research and innovation program under Marie-Sklodowska-Curie Grant No. 765267 (QUSCO). We  gratefully  acknowledge  useful  discussions  with Pr.~Christiane Koch, who suggested us to write this tutorial and provided advice throughout the project. We thank the PhD students J. Fischer, L. Kerber, M. Krauss, J. Langbehn, A. Marshall, and F. Roy for their helpful comments.

\bibliographystyle{apsrev}

\begin{thebibliography}{100}
\bibitem{comp4}
The reason is that
a minimizing sequence
of piecewise constant controls may fail to have a cluster point in the class (intervals on which the control is constant may have zero duration at the limit).
For this reason, optimal control problems are set for a larger class of controls in which optimal solutions can be shown to exist (under rather general conditions that we discuss in the main text) and can be characterized through suitable necessary conditions. Such conditions often allow to ensure that the optimal controls are indeed piecewise continuous.
We should stress, however, that there exist some examples, such as the Fuller problem~\cite{fuller}, for which the optimal control law has infinitely many switches between two values of $U$ in finite time. Such a control law is not piecewise continuous, but belongs to $L^\infty([0,T],U)$.
\bibitem{comp5} Notice that even if $M$ is here a manifold with a boundary $\partial M$, the forward time dynamics are well defined because the admissible
 velocities at any point of $\partial M$
 are either directed towards the interior of $M$ or are tangent to $\partial M$~\cite{altafini2004}.
\bibitem{comp6} Note that this statement can be deduced at once by taking any converging subsequence of a minimizing sequence.
\end{thebibliography}

\end{document}